\documentclass[11pt]{article}

\usepackage[margin=2.5cm]{geometry}
\usepackage{authblk}

\usepackage{amssymb,amsmath,amsthm,amsfonts,thmtools}
\usepackage{color}

\usepackage{yfonts}
\usepackage{mathrsfs}

\usepackage{etoolbox} %For AtBeginEnvironment

\usepackage{graphicx}
\usepackage[export]{adjustbox}
\usepackage{xspace}

\usepackage{times}
\usepackage{verbatim}
\usepackage{enumitem}
\usepackage{graphicx,wrapfig,lipsum}
\usepackage[font=small]{caption}
\usepackage{subcaption}
\usepackage[english]{babel}
\usepackage[utf8]{inputenc}

\usepackage[noend]{algpseudocode}
\usepackage{url}
\usepackage[all]{xy}
\usepackage{rotating}
\usepackage{ifpdf}
\usepackage{tcolorbox}
\usepackage{lipsum}

\usepackage{mdframed}             % \begin{framed}
\usepackage{xifthen} % \ifthenelse \isempty
\usepackage{mathtools}  % \coloneq

\DeclareMathAlphabet\mathbfcal{OMS}{cmsy}{b}{n}

\usepackage{comment}
\usepackage[T1]{fontenc}

\usepackage[hidelinks]{hyperref} 

\newcommand{\calI}{\ensuremath{\mathcal I}\xspace}

\newcommand{\calK}{\ensuremath{\mathcal K}\xspace}

\colorlet{darkgreen}{green!45!black}

\newtheorem{theorem}{Theorem}[section]
\newtheorem{lemma}{Lemma}[section]
\newtheorem{corollary}[lemma]{Corollary}

\newtheorem{observation}[lemma]{Observation}

\newcounter{case}
\newcounter{subcase}[case]
\newcounter{subsubcase}[subcase]
\newenvironment{case}[1][\unskip]{
\paragraph{Case \thecase:} #1.
}

\newenvironment{subcase}[1][\unskip]{
\paragraph{Case \thecase.\thesubcase:} #1.
}

\newenvironment{subsubcase}[1][\unskip]{
\paragraph{Case \thecase.\thesubcase.\thesubsubcase:} #1.
}

\AtBeginEnvironment{case}{\stepcounter{case}}
\AtBeginEnvironment{proof}{\setcounter{case}{0}}
\AtBeginEnvironment{subcase}{\stepcounter{subcase}}
\AtBeginEnvironment{subsubcase}{\stepcounter{subsubcase}}
\newcommand{\ignore}[1]{}

\newcommand{\margincomment}[2]%
{\marginpar{\footnotesize\raggedright {\color{red}#1}: #2}}

\newcommand{\myparagraph}[1]{{\medskip\noindent\textbf{#1}}}

\newcommand{\braced}[1]{{ \left\{ {#1} \right\} }}

\newcommand{\angled}[1]{{ \langle {#1} \rangle }}
\newcommand{\brackd}[1]{{ \left[ {#1} \right] }}

\newcommand{\parend}[1]{{ \left({#1} \right) }}
\newcommand{\barred}[1]{{ \left| {#1} \right| }}
\newcommand{\suchthat}{{\;:\;}}
\newcommand{\Keys}{\calK}

\newcommand{\threeWCST}{\textsc{3wcst}\xspace}
\newcommand{\twoWCST}{\textsc{2wcst}\xspace}

\DeclareMathOperator{\cost}{cost}
\DeclareMathOperator{\depth}{depth}
\newcommand{\highlambda}{\lambda^+}
\newcommand{\lowlambda}{\lambda^-}
\newcommand{\lowgamma}{\gamma^-}
\newcommand{\highgamma}{\gamma^+}

\newcommand{\equalstest}[1]{\angled{=#1}}
\newcommand{\lessthantest}[1]{\angled{<#1}}

\newcommand{\half}{\textstyle{\frac{1}{2}}}
\newcommand{\onethird}{\tfrac{1}{3}}

\newcommand{\threefourths}{\tfrac{3}{4}}
\newcommand{\onefourth}{\tfrac{1}{4}}
\newcommand{\twofifths}{\tfrac{2}{5}}
\newcommand{\onesixth}{\tfrac{1}{6}}
\newcommand{\threesevenths}{\tfrac{3}{7}}

\newcommand{\fournineths}{\tfrac{4}{9}}

\newcommand{\ListLengths}{\setlength{\itemsep}{0ex}\setlength{\topsep}{1ex}\setlength{\partopsep}{0ex}}

\graphicspath{{./}}

\usepackage{array}
\usepackage[title]{appendix}
\usepackage{graphicx}
\usepackage{tikz}
\usetikzlibrary{automata,positioning} \usepackage{dsfont} \usepackage{xspace}
\usepackage[nothing]{algorithm}
\usepackage{algorithmicx}
\usepackage[noend]{algpseudocode}
\algblockdefx[protocol]{Protocol}{EndProtocol}{\textbf{Protocol}\ }{}

\floatname{algorithm}{Pseudocode}

\title{Structural Properties of Search Trees with 2-way Comparisons\thanks{Research supported by NSF grant CCF-2153723.}}

\author[$\dagger$]{Sunny Atalig}
\author[$\dagger$]{Marek Chrobak}
\author[$\ddagger$]{Erfan Mousavian}
\author[$\S$]{Jiri Sgall}
\author[$\S$]{Pavel Vesely}

\affil[$\dagger$]{University of California at Riverside, USA}
\affil[$\ddagger$]{Sharif University of Technology, Tehran, Iran}
\affil[$\S$]{Charles University, Prague, Czech Republic}

\begin{document}

\maketitle

\begin{abstract}
Optimal 3-way comparison search trees (\threeWCST's) can be computed using
standard dynamic programming in time $O(n^3)$, and this can be further improved
to $O(n^2)$ by taking advantage of the Monge property.
In contrast, the fastest algorithm in the literature
for computing optimal 2-way comparison search
trees (\twoWCST's) runs in time $O(n^4)$.  To shed light on this discrepancy,
we study structure properties of \twoWCST's. On one hand, we show
some new threshold bounds involving key weights that can be helpful in
deciding which type of comparison should be at the root of the optimal tree.
On the other hand, we also show that the standard techniques for
speeding up dynamic programming (the Monge property / quadrangle inequality)
do not apply to \twoWCST's.
\end{abstract}

%%%%%%%%%%%%%%%%%%%%%%%%%%%%%%%%%%%%%%%%%%%%%%%%%%%%%%%%%%%%%%%%%%%%%%%%%%%%%%
%%%%%%%%%%%%%%%%%%%%%%%%%%%%%%%%%%%%%%%%%%%%%%%%%%%%%%%%%%%%%%%%%%%%%%%%%%%%%%
%%%%%%%%%%%%%%%%%%%%%%%%%%%%%%%%%%%%%%%%%%%%%%%%%%%%%%%%%%%%%%%%%%%%%%%%%%%%%%

\section{Introduction}
\label{sec: introduction}

%%%%%%%%%%%%%%%%%%%%%%%%%%%%%%%%%%%%%%%%%%%%%%%%%%%%%%%%%%%%%%%%%%%%%%%%%%%%%%
%%%%%%%%%%%%%%%%%%%%%%%%%%%%%%%%%%%%%%%%%%%%%%%%%%%%%%%%%%%%%%%%%%%%%%%%%%%%%%

%\section{Introduction}
%\label{sec: introduction}
%\input{01_introduction.tex}

%%%%%%%%%%%%%%%%%%%%%%%%%%%%%%%%%%%%%%%%%%%%%%%%%%%%%%%%%%%%%%%%%%%%%%%%%%%%%%
%%%%%%%%%%%%%%%%%%%%%%%%%%%%%%%%%%%%%%%%%%%%%%%%%%%%%%%%%%%%%%%%%%%%%%%%%%%%%%

We have a linearly ordered set $\Keys$ of $n$ keys.
We consider \emph{two-way comparison search trees} (\twoWCST's), 
namely decision trees whose internal nodes
represent equal-to comparisons or less-than comparisons to keys.
Each internal node in a \twoWCST has two children, one corresponding to
the ``yes" answer and the other to the ``no" answer.  Each leaf represents one key from $\Keys$.

Fix some \twoWCST $T$. When searching in $T$, we are given a query value $q$ that we assume
to be one of the keys%
\footnote{We remark that this is sometimes referred to as the
\emph{successful-query model}. In the \emph{general model}, a query $q$ can be
any value, and if $q\notin\Keys$ 
then the search for $q$ must end in the leaf representing the
inter-key interval containing $q$. Algorithms
for the succcessful-query model typically extend naturally to the
general model without affecting the running time.
}, 
that is $q \in \Keys$. 
We search for $q$ in $T$ as follows: starting at the root, $q$ is 
compared to the key in the current node (an equal-to or less-than test),
and the search then proceeds to its child according to the result of the comparison.
For each $q$, the search must end up in the correct leaf, namely
if $q = k$ then the search must end up in the leaf labelled $k$. 

Each key $k\in\Keys$ is assigned some weight $w_k \ge 0$. We define the cost
of $T$ by $\cost(T) = \sum_{k\in\Keys}^n w_k\cdot \depth(k)$, where
the depth  of a key $k$ is the distance (that is, the number of comparisons) from the leaf of $k$ to the root.
This cost can be also comptued as follows: Let $w(v)$ be the
total weight of the leaves in the subtree rooted at a node $v$.
Then the cost of $T$ is $\cost(T) = \sum_{v\in T'} w(v)$, where
$T'$ is the set of the internal (comparison) nodes of $T$.
The objective is to compute a \twoWCST $T$ for the given instance that minimizes $\cost(T)$.

\begin{figure}[ht]
\begin{center}
\includegraphics{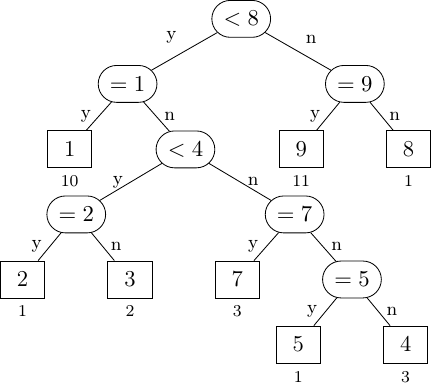} 
\caption{An example of a binary search tree, with
items $1,2,3,4,5,7,8,9$. If the weights of
these items are, respectively, $10,1,2,  3,1,3,  1,11$,
the the cost of the tree is
$10\cdot 2 + 1 \cdot 4 + 2\cdot 4 + 3 \cdot 5 + 1 \cdot 5 
	+ 3\cdot 4 + 1 \cdot 2 + 11 \cdot 2 = 88$.
Adding the weights of the internal leaves, layer by layer,
we get $32 + 20 + 12 + 10 + 3  + 7 + 4  = 88$.
}
\end{center}
\end{figure}

The problem of computing optimal search trees, or computing optimal decision trees in general,
is one of classical  problems in the theory of algorithms. Initial research in this area, in 1960's and 70's, 
focussed on 3-way comparison trees (\threeWCST). Optimal {\threeWCST}'s
can be computed in time $O(n^3)$ using standard dynamic programming,
and Knuth~\cite{Knuth1971} and Yao~\cite{yao_efficient_dynamic_programming_80} developed a faster, $O(n^2)$-time 
algorithm by showing that such optimal trees satisfy what is now referred
to as the quadrangle inequality or the Monge property. 

The standard dynamic programming approach does not apply to \twoWCST's.
After earlier attempts in~\cite{Spuler1994Paper,Spuler1994Thesis}, that 
were shown to be flawed~\cite{chrobak_huang_2022},
Anderson et al.~\cite{Anderson2002} developed an algorithm with running time $O(n^4)$.
Their algorithm was later simplified in~\cite{chrobak_simple_2021}.
(See~\cite{chrobak_simple_2021} also for a more extensive summary of the literature and related results.)

Thus, in spite of apparent similarity between these two problems and past research efforts,
the running time of the best algorithm for optimal \twoWCST's is still two
orders of magnitude slower than the one for {\threeWCST}'s.
To shed light on this discrepancy, in this paper we study structure properties of \twoWCST's.
Here is a summary of our results:
\begin{itemize}
\item An algorithm for designing optimal \twoWCST's needs to determine, in particular,
whether the root node should use the equal-to test or the less-than test.
Let $W$ be the total key weight.
The intuition is that there should exist thresholds $\lowlambda$ and $\highlambda$
on key weights with the following properties:
(i) if the heaviest key weight exceeds $\highlambda\cdot W$
then some optimum tree starts with the equal-to test in the root, and
(ii) if the heaviest key weight is below $\lowlambda\cdot W$
then each optimum tree starts with the less-than test.
Anderson et al.~\cite{Anderson2002} confirmed this intuition and 
proved that $\highlambda\in [\threesevenths,\fournineths]$ and $\lowlambda = \onefourth$. 
This was recently refined in a companion paper~\cite{atalig_chrobak_tight_threshold_23}
that shows that $\highlambda  = \threesevenths$.
In Section~\ref{sec: threshold results} we provide additional threshold bounds,
one involving \emph{two} largest weights, the other involving so-called
\emph{side weights}, and establish some tightness results.
\item The dynamic programming formulation for computing
optimal {\twoWCST}'s is based on a recurrence (see Section~\ref{sec: preliminaries})
involving $O(n^3)$ sub-problems, resulting in an $O(n^4)$-time algorithm~\cite{Anderson2002,chrobak_simple_2021}. 
For many weight sequences the number of sub-problems that need to be considered is
substantially smaller than $O(n^3)$, but there are also sequences for which no methods
are known to reduce the number of sub-problems to $o(n^3)$. Such
sequences involve geometrically increasing weights, say $1,\gamma,\gamma^2,...$.
In Section~\ref{sec: structure results for geometric sequences} we give a
threshold result for geometric sequences, by proving that
for $\gamma < \lowgamma = (\sqrt{5}-1)/2 \approx 0.618$ there is an optimal
tree starting with the equal-to test in the root.
(Note that $\lowgamma$ corresponds to the weight threshold $1-\lowgamma \approx 0.382$,
which is smaller than $\highlambda = \threesevenths \approx 0.429$.)
\item As discussed earlier, the $O(n^2)$-time algorithm
for computing optimal {\threeWCST}'s uses a dynamic-programming speed-up technique based on 
the so-called quadrangle inequality~\cite{Knuth1971,yao_efficient_dynamic_programming_80}
(which is closely related to and sometimes referred to as the Monge property).
Section~\ref{sec: non-monotonicity of dynamic-programming minimizers} focuses
on this technique. With a series of counter-examples, we show that
neither the quadrangle inequality nor a number of other (possibly useful) related
structural properties apply to {\twoWCST}'s.
\item In Section~\ref{sec: algorithms for bounded weights}
we show that faster than $O(n^4)$ algorithms for {\twoWCST}'s
can be achieved if, roughly,
the weight sequence does not grow too fast.
If the weights are from the range $[1,R]$, for some integer $R$,
then there is a simple $O(n^3)$-time algorithm for the case when $R$ is
constant.
For arbitrary $R$, we give an algorithm with running time $O(n^3\log(nR))$,
which beats the running time $O(n^4)$ on instances where
the weight sequence grows slower than geometrically.
	
\end{itemize}

%%%%%%%%%%%%%%%%%%%%%%%%%%%%%%%%%%%%%%%%%%%%%%%%%%%%%%%%%%%%%%%%%%%%%%%%%%%%%%
%%%%%%%%%%%%%%%%%%%%%%%%%%%%%%%%%%%%%%%%%%%%%%%%%%%%%%%%%%%%%%%%%%%%%%%%%%%%%%
%%%%%%%%%%%%%%%%%%%%%%%%%%%%%%%%%%%%%%%%%%%%%%%%%%%%%%%%%%%%%%%%%%%%%%%%%%%%%%

\section{Preliminaries}
\label{sec: preliminaries}

%%%%%%%%%%%%%%%%%%%%%%%%%%%%%%%%%%%%%%%%%%%%%%%%%%%%%%%%%%%%%%%%%%%%%%%%%%%%%%
%%%%%%%%%%%%%%%%%%%%%%%%%%%%%%%%%%%%%%%%%%%%%%%%%%%%%%%%%%%%%%%%%%%%%%%%%%%%%%

%\section{Preliminaries}
%\label{sec: preliminaries}
%\input{02_preliminaries.tex}

%%%%%%%%%%%%%%%%%%%%%%%%%%%%%%%%%%%%%%%%%%%%%%%%%%%%%%%%%%%%%%%%%%%%%%%%%%%%%%
%%%%%%%%%%%%%%%%%%%%%%%%%%%%%%%%%%%%%%%%%%%%%%%%%%%%%%%%%%%%%%%%%%%%%%%%%%%%%%

\paragraph{Notation.} 
Recall that $\Keys$ denotes the set of keys. We will typically use letters $i,j,k,...$ for keys.
By $w_k$ we denote the weight of key $k\in \Keys$.
The specific key values are not relevant --- in fact one can assume that $\Keys = \braced{1,2,...,n}$.
We can then represent a problem instances simply by its weight sequence $w_1 , \ldots, w_n$.
For simplicity, in most of this paper we assume that all weights
are different\footnote{The extension to arbitrary weights is actually
not trivial, but it can be done without affecting the time complexity.
See~\cite{chrobak_simple_2021}, for example, for the discussion of this issue.}.

By $W = \sum_{k\in\Keys} w_k$ we denote the \emph{total weight} of the instance, 
and we say that an instance is \emph{normalized} if $W=1$. If $T$ is an \twoWCST, 
then its weight $w\parend{T}$ is the total weight of the instance it handles (i.e. the sum of all weights of its leaves). 
If $v$ is a node in $T$, its weight $w\parend{v}$ is the weight of the sub-tree rooted at $v$.

For convenience, we denote an equal-to test to key $k$ as $\equalstest{k}$
and less-than test as $\lessthantest{k}$. 
Furthermore, the ``left'' and ``right'' branches of a comparison test 
refer to the ``yes'' and ''no'' outcomes respectively.
To simplify some proofs, we will sometimes label the outcomes with relation symbols
``$=$'', ``''$\neq$'', ``$<$'', etc.

%%%%%%%%%%%%%%%%%%%%%%%

\paragraph{Side weights.} 
The concept of side-weights was defined by Anderson et al.~\cite{Anderson2002},
who also proved the three lemmas that follow.

	Let \(v\) be a node in a \twoWCST $T$. Then the \emph{side-weight} of \(v\) is 
	\[ sw\parend{v} = \begin{cases}
		0 , & v \text{ is a leaf}, \\
		w, &v \text{ is an equal-to test on a key of weight } w , \\
		\min\braced{w\parend{T_L} , w\parend{T_R}} & v \text{ is a less-than test with sub-trees } T_L, T_R .
	\end{cases}\]

\begin{lemma} \label{lemma: side-weight monoticity}
	\emph{\cite{Anderson2002}}
	Let \(T\) be an optimal \twoWCST. Then \(sw\parend{u} \ge sw\parend{v}\) if \(u\) is a parent of \(v\).
\end{lemma}

\begin{lemma} \label{lemma: RMLK property}
	\emph{\cite{Anderson2002}}
	For \(n > 2\), if an optimal \twoWCST for any instance of \(n\) keys is rooted at an equal-to test on \(i\), 
	then \(i\) is a key of maximum-weight.
\end{lemma}

%%%%%%%%%%%%%%%%%%%%%%%

\paragraph{Dynamic programming.}
We now show a dynamic programming formula for the optimal cost of \twoWCST's.
In this formulation we assume that the key weights are different. This formulation
involves $O(n^3)$ subproblems, and the cost of each involves a minimization over
$O(n)$ smaller subproblems. As shown in~\cite{chrobak_simple_2021}, with appropriate
data structures this leads to an $O(n^4)$-time algorithm, that also applies to instances with
equal weight keys and to the general model with unsuccessful queries.

Assume that $\Keys = \braced{1,2,...,n}$.
Let $a_1,a_2,...,a_n$ be the permutation of $1,2,...,n$, such that 
$w_{a_1} < w_{a_2} < ... < w_{a_n}$. For $1\le i \le j \le n$ and $1\le h \le n$ we define a sub-instance
\begin{eqnarray*}
	I^h_{i,j} &=& \braced{i,...,j} \cap \braced{a_1,...,a_h}.
\end{eqnarray*}
Thus $I^h_{i,j}$ consists of all keys among $i,i+1,...,j$ whose weight is at most $w_{a_h}$.
Let $w^h_{i,j}$ be the total weight of the keys in $I^h_{i,j}$.

By $C^h_{i,j}$ we denote the minimum cost of the \twoWCST for
$I^h_{i,j}$. Our goal is to compute $C^n_{1,n}$ --  the optimal cost of the \twoWCST for the whole instance.

We now give a recurrence for $C^h_{i,j}$. If $|I^h_{i,j}| = 1$ then
$C^h_{i,j} = 0$. (This corresponds to a leaf $k$, for $I^h_{i,j} = \braced{k}$.)
Assume that $|I^h_{i,j}|\ge  2$ (in particular, $i < j$).
If $a_h\notin I^h_{i,j}$ then $C^h_{i,j} = C^{h-1}_{i,j}$,
because in this case $I^h_{i,j} = I^{h-1}_{i,j}$. If $a_h\in I^h_{i,j}$ then
compute $C^h_{i,j}$ as follows:

\begin{align}
	C^h_{i,j} \;&=\; w^h_{i,j} + 
	\min \braced{ C^{h-1}_{i,j} , S^h_{i,j} },  \quad\textrm{where}  
	\label{eqn: formula for C}
	\\
	S^h_{i,j} \;&=\; 
	\min_{l = i,...,j-1}\braced{ C^h_{i,l} + C^h_{l+1,j}}.
	\label{eqn: formula for S} 
\end{align}
The two choices in the minimum formula for $C_{i,j}^h$
correspond to creating an internal node, either
the equal-to test
$\equalstest{{a_h}}$ or a less-than-test $\lessthantest{{l+1}}$.
We refer to the index (key) $l$ in the formula of $S_{i, j}^h$ to be a ``cut'' or ``cut-point'', 
and define $C_{i, j}^h \parend{l} = C^h_{i,l} + C^h_{l+1,j}$.

%%%%%%%%%%%%%%%%%%%%%%%%%%%%%%%%%%%%%%%%%%%%%%%%%%%%%%%%%%%%%%%%%%%%%%%%%%%%%%
%%%%%%%%%%%%%%%%%%%%%%%%%%%%%%%%%%%%%%%%%%%%%%%%%%%%%%%%%%%%%%%%%%%%%%%%%%%%%%
%%%%%%%%%%%%%%%%%%%%%%%%%%%%%%%%%%%%%%%%%%%%%%%%%%%%%%%%%%%%%%%%%%%%%%%%%%%%%%

\section{Threshold Results}
\label{sec: threshold results}

%%%%%%%%%%%%%%%%%%%%%%%%%%%%%%%%%%%%%%%%%%%%%%%%%%%%%%%%%%%%%%%%%%%%%%%%%%%%%%
%%%%%%%%%%%%%%%%%%%%%%%%%%%%%%%%%%%%%%%%%%%%%%%%%%%%%%%%%%%%%%%%%%%%%%%%%%%%%%

%\section{Threshold Results}
%\label{sec: threshold results}
%\input{03_threshold_results.tex}

%%%%%%%%%%%%%%%%%%%%%%%%%%%%%%%%%%%%%%%%%%%%%%%%%%%%%%%%%%%%%%%%%%%%%%%%%%%%%%
%%%%%%%%%%%%%%%%%%%%%%%%%%%%%%%%%%%%%%%%%%%%%%%%%%%%%%%%%%%%%%%%%%%%%%%%%%%%%%

The following results were given in~\cite{Anderson2002,atalig_chrobak_tight_threshold_23}:

\begin{theorem}\label{theorem: lambda-}
\emph{\cite{Anderson2002}}
	Let \(\lambda^-\) be the largest value such that for all \(p \in \parend{0, \lambda^-}\) and all normalized instances with max key-weight \(p\), 
	there exists no optimal \twoWCST rooted at an equal-to test. Then $\lambda^- = \onefourth$.
\end{theorem}

\begin{theorem}\label{theorem: lambda+ lower bound}
\emph{\cite{Anderson2002,atalig_chrobak_tight_threshold_23}}
	Let \(\lambda^+\) be the smallest value such that for all \(p \in \parend{\lambda^+, 1}\) and all normalized instances with max key-weight \(p\), 
	there exists an optimal \twoWCST rooted at an equal-to test. Then $\highlambda = \threesevenths$.
\end{theorem}

\begin{lemma} \label{lemma: Split Tree property}
\emph{\cite{Anderson2002}}
	Let $\Keys$ be an instance of keys and $\Keys_1$ and $\Keys_2$ be two complementary sub-sequences of $\Keys$. 
	Let $T$, $T_1$, and $T_2$ 
	be optimal \twoWCST{s} for $\Keys$, $\Keys_1$, $\Keys_2$ respectively. 
	Then $\cost\parend{T_1} + \cost\parend{T_2} \le \cost\parend{T}$.
\end{lemma}

The last lemma is in fact quite obvious: One can think of $T$ as a redundant search tree that handles both instances
$\Keys_1$, $\Keys_2$, and the search costs for these instances add up to the cost of $T$. Removing redundant comparisons
can only decrease cost.

%%%%%%%%%%%%%%%%%%%%%%%%%%%%%%%%%%%%%%%%%%%%%%%%%%%%%%%%%%%%%%%%%%%%%%%%%%%%%%
%%%%%%%%%%%%%%%%%%%%%%%%%%%%%%%%%%%%%%%%%%%%%%%%%%%%%%%%%%%%%%%%%%%%%%%%%%%%%%

\subsection{Equal-to Test Threshold for Two Heaviest Keys}
\label{sec: equality-test threshold for two heaviest keys}

%%%%%%%%%%%%%%%%%%%%%%%%%%%%%%%%%%%%%%%%%%%%%%%%%%%%%%%%%%%%%%%%%%%%%%%%%%%%%%
%%%%%%%%%%%%%%%%%%%%%%%%%%%%%%%%%%%%%%%%%%%%%%%%%%%%%%%%%%%%%%%%%%%%%%%%%%%%%%

%\subsection{Equality-Test Threshold for Two Heaviest Keys}
%\label{sec: equality-test threshold for two heaviest keys}
%\input{03-1_equality-test_threshold_for_two_heaviest_keys.tex}

%%%%%%%%%%%%%%%%%%%%%%%%%%%%%%%%%%%%%%%%%%%%%%%%%%%%%%%%%%%%%%%%%%%%%%%%%%%%%%
%%%%%%%%%%%%%%%%%%%%%%%%%%%%%%%%%%%%%%%%%%%%%%%%%%%%%%%%%%%%%%%%%%%%%%%%%%%%%%

In this section we provide threshold results that depend on the two heaviest weights, that we denote by $\alpha$ and $\beta$, with $\alpha\ge\beta$. 
By the threshold bound in~\cite{atalig_chrobak_tight_threshold_23},
if $\alpha\ge\threesevenths W$ then there is an optimal \twoWCST with an equal-to
test in the root. The theorem below shows that this threshold value depends on $\beta$.
In fact, $\beta \ge \onethird W$  already suffices for an optimal tree with an equal-to test at the root to exist.

\begin{theorem}\label{theorem: Second Heaviest Key Equality Threshold}
	Consider an instance of  $n \ge 2 $ keys with total weight  $W $. Let  $\alpha $ and  $\beta $ be the first and 
	second heaviest weights respectively, where
	\begin{equation}\label{equation: 2nd Key Threshold}
		2\alpha + \beta \ge W.
	\end{equation}
	Then there exists an optimal \twoWCST\ rooted at an equal-to test. 
\end{theorem}

Recall that, by Lemma~\ref{lemma: RMLK property}, if an optimal tree's root is an equal-to test then
we can assume that this test is to the heaviest key.

\begin{proof}
	We prove the theorem by induction. The base cases  $n = 2, 3 $ are trivial, so let  $n \ge 4 $ and assume that the theorem holds for fewer than $n$ keys.
	Let  $T $ be a tree for an  $n $-key instance rooted at  $\lessthantest{r} $.
	Suppose that $2\alpha + \beta \ge w\parend{T} $. 
	We show that we can then convert $T$, without increasing cost, into a tree $T'$ with an equal-to test in the root.
	
	We can assume that  $2 < r \le n-1 $ since otherwise we can replace the  less-than test by an  equal-to test.
	So we have at least 2 nodes on each side. 
	Let  $k_\alpha $ and  $k_\beta $ be the keys with weights  $\alpha $ and  $\beta $. 
	For simplicity, we assume in the proof that keys $k_\alpha$ and $k_\beta$ are unique.
	The argument easily extends to the case when they are not.
	For example, in some cases when the root of a subtree containing $k_\beta$ but not $k_\alpha$ is an equal-to test,
	knowing that this test must be to the heaviest key in the subtree, we let this root be  $\equalstest{k_\beta}$.
	If this root is some other key of weight $\beta$, use this key instead in the argument, and the subsequent weight estimates do not change.
	
	We now break into two cases based on the location of these keys.
%	\begin{enumerate}[leftmargin=*, label=Case \arabic*:]
%		\item  $k_\alpha $ and  $k_\beta $ are on the same side w.r.t.\  $r $.
%		\item  $k_\alpha $ and  $k_\beta $ are on opposite sides w.r.t.\  $r $.
%	\end{enumerate}
	
	\begin{case}[$k_\alpha $ and  $k_\beta $ are on the same side w.r.t.\  $r $]\label{case 1: 2nd Heavy Key}
		We can assume that  $k_\alpha, k_\beta < r $, by symmetry.
		The left tree  $T_L $ handles less than  $n $ nodes, and since  $w\parend{T_L} \le w\parend{T} $, we have  $2\alpha + \beta \ge w\parend{T} \ge w\parend{T_L} $. 
		We can apply our inductive assumption and replace  $w\parend{T_L} $ with an equivalent-cost tree rooted at  $\equalstest{k_\alpha} $. 
		Let $T_1 $ be the right sub-tree of  $\equalstest{k_\alpha} $ and  $T_2 $ be the right sub-tree of  $\lessthantest{r} $. 
		Note that  $T_1 $ must contain  $\beta $ and  $w\parend{T_1} \ge \beta $. 
		To obtain $T'$, modify the tree by rotating  $\equalstest{k_\alpha} $ to the root, which moves leaf $k_\alpha $ up by  $1 $ and  $T_2 $ down by  $1 $. 
		The change in cost is
		
		\begin{align*}
			\cost(T')-\cost(T) \;&=\; -\alpha + w\parend{T_2} \\
			&=\; -\alpha + \parend{w\parend{T} - \alpha - w\parend{T_1}}  &&\alpha +w\parend{T_1} + w\parend{T_2} = w\parend{T} \\
			&\le\; w\parend{T} -2\alpha - \beta &&\beta \le w\parend{T_1} \\
			&\le\; 0 && \text{by~(\ref{equation: 2nd Key Threshold})}
		\end{align*}
	\end{case}
	
	\begin{case}[$k_\alpha $ and  $k_\beta $ are on opposite sides w.r.t.\  $r $]\label{case 2: 2nd Heavy Key}
		We can assume that  $k_\alpha < r \le k_\beta $, by symmetry. 
		Since $T_L$ doesn't contain  $k_\beta $, we have  $w\parend{T_L} \le w\parend{T} - \beta $. 
		By~(\ref{equation: 2nd Key Threshold}), we have  $\alpha \ge \half \parend{w\parend{T} - \beta} \ge \half w\parend{T_L} $. 
		So by  $\highlambda < \half$ (see~\cite{atalig_chrobak_tight_threshold_23}), 
		we can replace  $T_L $ with an equivalent cost tree rooted at  $\equalstest{k_\alpha} $. Again let  $T_1 $ be the right sub-tree of  $\equalstest{k_\alpha} $ and  $T_2 $ be the right sub-tree of  $\lessthantest{r} $. We now break into two sub-cases based on the root of  $T_2 $.
%		\begin{enumerate}[leftmargin=*, label=Case \thecase.\arabic*:]
%			\item  $T_2 $ is rooted at an  equal-to test.
%			\item  $T_2 $ is rooted at an  less-than test.
%		\end{enumerate}

		\begin{subcase}[$T_2 $ is rooted at an  equal-to test]\label{case 2.1: 2nd Heavy Key}
			If  $T_2 $ is rooted at an  equal-to test and $k_\alpha$ is not in $T_2$, 
			then this test is to key $k_\beta$. Let  $T_3 $ be the right sub-tree of  $\equalstest{k_\beta} $. 
			Then consider the tree $T'$ rooted at  $\equalstest{k_\alpha} $, followed by  $\equalstest{k_\beta} $, then followed by  $\lessthantest{r}$.  
			In $T'$, subtrees $T_1 $ and  $T_3 $ go down by  $1 $ and  leaf $k_\alpha $ goes up by  $1 $. The change in cost is
			
			\begin{align*}
				\cost(T') - \cost(T) \;&=\; - \alpha + w\parend{T_1} + w\parend{T_3} \\
				&=\; - \alpha + \parend{w\parend{T} - \alpha - \beta}  && \alpha + \beta + w\parend{T_1} + w\parend{T_3} = w\parend{T}\\
				&=\; w\parend{T} -2\alpha - \beta \\
				&\le\; 0
			\end{align*}
		\end{subcase}
		
		\begin{subcase}[$T_2 $ is rooted at an  less-than test]\label{case 2.2: 2nd Heavy Key}
			Let  $T_2 $ be rooted at  $\lessthantest{i}$, with  $T_3 $ and  $T_4 $ being the left and right branches respectively. 
			We may assume  $w(T_4) < \beta $, as otherwise we can do an identical rotation as in Case \ref{case 2.1: 2nd Heavy Key}. 
			This implies that  $T_3 $ must contain  $k_\beta $ and  $w\parend{T_3} \ge \beta $. 
			We can assume that $T_3 $ is not a leaf, for otherwise we can replace $\lessthantest{i}$ by an equal-to test,
			reducing it to Case~\ref{case 2.1: 2nd Heavy Key}.
			We now break into two more sub-cases based on the root of  $T_3 $.
%			\begin{enumerate}[leftmargin=*,label=Case \thesubcase.\arabic*:]
%				\item  $T_3 $ is rooted at an  equal-to test.
%				\item  $T_3 $ is rooted at an  less-than test.
%			\end{enumerate}
			
			\begin{subsubcase}[$T_3 $ is rooted at an  equal-to test]\label{case2.2.1} %\footnote{SA: I believe you could also handle this case by noting that  $T_4 \ge \beta $ by a rotation, a contradiction.}
				Let  $T_3 $ be rooted at an  equal-to test. Since  $T_3 $ contains  $k_\beta $ but not  $k_\alpha $,  $T_3 $ is rooted at  $\equalstest{k_\beta} $.
				But then, by Lemma~\ref{lemma: side-weight monoticity}, we have that $w(T_4) \ge \beta$, contradicting $w(T_4) < \beta$.	
				So this case is in fact not possible.			
			\end{subsubcase}
			
			\begin{subsubcase}[$T_3 $ is rooted at an  less-than test]\label{case 2.2.2: 2nd Heavy Key}
				Let  $T_3 $ be rooted at  $\lessthantest{j} $ with sub-trees  $T_5 $ and  $T_6 $. Then consider the tree $T'$ rooted at  $\equalstest{k_\alpha} $, 
				followed by  $\lessthantest{j} $, with  $\lessthantest{r} $ and  $\lessthantest{i} $ on the left and right branches respectively.  
				Leaf $k_\alpha $ goes up by  $1 $ and  $T_1 $ and  $T_4 $ go down by 1. The change in cost is
				
				\begin{align*}
					\cost(T') - \cost(T) \;&=\; -\alpha + w\parend{T_1} + w\parend{T_4} \\
					&=\; - \alpha + \parend{w\parend{T} - \alpha - w\parend{T_5} - w\parend{T_6}}  \\ 
					&\le\; w\parend{T} - 2\alpha - \beta  && k_\beta \text{ is in } T_5 \text{ or } T_6  \\
					&\le\; 0 &&\qedhere
				\end{align*}
				
			\end{subsubcase}

		\end{subcase}
	\end{case}

\end{proof}

%%%%%%%%%%%%%%%%%%%%%%%%%%%%%%%%%%%%%%%%%%%%%%%%%%%%%%%%%%%%%%%%%%%%%%%%%%%%%%
%%%%%%%%%%%%%%%%%%%%%%%%%%%%%%%%%%%%%%%%%%%%%%%%%%%%%%%%%%%%%%%%%%%%%%%%%%%%%%

\subsection{Threshold for Side-Weights}
\label{sec: threshold for side-weights}
%%%%%%%%%%%%%%%%%%%%%%%%%%%%%%%%%%%%%%%%%%%%%%%%%%%%%%%%%%%%%%%%%%%%%%%%%%%%%%
%
%\subsection{Threshold for Side-Weights}
%\label{sec: threshold for side-weights}
%\input{03-2_threshold_for_side-weights.tex}
%
%%%%%%%%%%%%%%%%%%%%%%%%%%%%%%%%%%%%%%%%%%%%%%%%%%%%%%%%%%%%%%%%%%%%%%%%%%%%%%

With small modification, the proof for $\lowlambda \ge \onefourth$ given by \cite{Anderson2002} can be extended to establish more general properties about side-weights,
given in 
Theorem~\ref{theorem: Side Weights Thresholds} below. Corollary~\ref{corollary: Second Heaviest Key Less-Than Thresshold} that follows
and the previous thresholds are discussed in Section~\ref{sec: tightness of thresholds}, 
and algorithmic implications of Corollary~\ref{corollary: First Side Weight Threshold} 
are discussed in Section~\ref{sec: bounding dynamic-programming minimizers by weight}.

Here, we borrow the convention of \emph{side branches} from \cite{Anderson2002}, corresponding to the side weight of a comparison node $v$. 
If $v$ is an equal-to test on key $k$, then its side branch is the leaf $k$. If $v$ is a less-than test, then its side branch is the sub-tree of 
least weight (break ties in weight arbitrarily). We then refer to the other branch of $v$ as the \emph{main branch}.

\begin{theorem}\label{theorem: Side Weights Thresholds}
	Let $T$ be an optimal tree for an instance of $n\ge 3$ keys, where $v_0$ is the root of $T$ and $v_1$ is the root of the main branch of $v_0$. Then
	\begin{equation}\label{equation: side weights thresholds}
		sw\parend{v_0} + sw\parend{v_1} \ge \tfrac{1}{2} w\parend{T}.
	\end{equation}
\end{theorem}

\begin{proof}
	 The proof is identical to the one for \(\lowlambda \ge \frac{1}{4}\) given by \cite{Anderson2002}, with minor modifications.
	 Let \(T\) be an optimal tree for some instance of \(n \ge 3\) keys, and let \(\parend{v_0 , v_1, \ldots, v_\ell}\) be the unique path along main branches from the root \(v_0\) to a leaf \(v_\ell\). Assume for contradiction that the theorem is false, that is
\begin{equation}
	 sw\parend{v_0} + sw\parend{v_1} < \half w\parend{T}
	 \label{eqn: side theorem false}
\end{equation}	 
	 Any given key is either \(v_\ell\) or occupies a unique side branch corresponding to some node \(v_i\). Therefore we can express the weight of \(T\) in terms of the side weights along of this path
	
	\[w\parend{T} = \parend{sw\parend{v_0} + sw\parend{v_1} + \cdots + sw\parend{v_{\ell - 1}}} + w\parend{v_\ell}.\]
	
	Note that while \(v_\ell\) is not a side branch, the node \(v_{\ell-1}\) can be replaced with \(\angled{= v_\ell}\), allowing us to treat it as a side branch. 
	By Lemma~\ref{lemma: side-weight monoticity}, we have \(w\parend{v_\ell}, sw\parend{v_{\ell - 1}} \le sw\parend{v_{\ell - 2}} \le \cdots \le sw\parend{v_{2}} \le sw\parend{v_1}\). We thus get the inequality
	
	\[w\parend{T} \le sw\parend{v_0} + \ell \cdot sw\parend{v_1}.\]
	
	If \(\ell \le 3\), we have \(w\parend{T} \le sw\parend{v_0} + sw\parend{v_1} + 2sw\parend{v_1} < w\parend{T}\), from~(\ref{eqn: side theorem false}), a contradiction. 
	
	So assume \(\ell \ge 4\). We'll now consider the first four nodes in the main branch, letting $T_i$ be the side branch of $v_i$ for $0 \le i \le 3$, $T_4$ be the main branch of $v_3$, and $k_i$ be the node tested by $v_i$. By Lemma~\ref{lemma: side-weight monoticity} we have $w\parend{T_0} \ge w\parend{T_1} \ge w\parend{T_2} \ge w\parend{T_3}$, and $w\parend{T_0} + w\parend{T_1} < \half w\parend{T}$ by our assumption for contradiction. Also note that if $v_i$ is a less-than test, then $k_j$ for $j > i$ either all precede or all succeed  $k_i$, depending on whether $v_i$'s main branch is the yes or no outcome.
	
	We'll now modify $T$ to get a new tree $T'$ of strictly lower cost. First re-label the keys $\braced{k_0 , k_1, k_2, k_3} = \braced{b_0 , b_1, b_2, b_3}$ so that $b_i \le b_j$ for $i < j$. Like in the proofs for bounding $\lambda^-$ and $\lambda^+$, we use one of the``inner cut-points'' $b_1$ or $b_2$ to split $T$ into a wider tree. We break into two cases:
	
%	\begin{enumerate}[leftmargin=*,label=Case \arabic*:]
%		\item  .
%		\item  .
%	\end{enumerate}
	
	\begin{case}[$b_1$ and $b_2$ both correspond to equal-to tests]
		In this case, let $T'$ be rooted at $\lessthantest{b_2}$, where the tests from $T$ on $b_0$ and $b_1$ are done on the left branch and likewise for $b_2$ and $b_3$ on the right branch. Do these tests in the same order as in the original tree. Because we've introduced a new test into the tree, one $T_i$ must fracture, and we can show that that sub-tree must be $T_3$ or $T_4$. In this case the sub-trees corresponding to $b_1$ and $b_2$ are leaves and can't split, so the sub-tree must correspond to an less-than test on $b_0$ or $b_3$, the least or greatest key. If the fractured tree $T_i$ corresponds to $b_0$ and $b_0 = v_i$ for $i < 3$, then this implies that the main branch of $v_i$ corresponds the yes or $\angled{\ge}$ outcome and $T_i$ is the opposite outcome. Then $T_i$ does not contain any other keys $v_j$ for $j \ne i$ and can't be split by $\lessthantest{b_2}$. The case is symmetric if $T_i$ corresponds to the test on $b_3$.
		
		So only $T_3$ or $T_4$ split, which does not increase cost by Lemma~\ref{lemma: Split Tree property} (or if both are leaves, then we can remove test $v_3$ or $v_4$ from $T'$, only decreasing cost). By doing tests in order of the original tree, $T_3$ and $T_4$ (or their fractured components) must be at the bottom of $T'$, moving up by 1 to depth 3, and $v_0$ will be right below $\lessthantest{b_2}$, pushing down $T_0$ by 1 to depth 2. Depending on which what branches they reside, $T_1$ and $T_2$ don't change in depth, or $T_1$ goes down by 1 and $T_2$ goes up by 1 (the latter being the worse case since $w\parend{T_1} \ge w\parend{T_2}$). Summarizing, our change in cost is
		
		\begin{align*}
			\cost(T')-\cost(T) \;&\le\; w\parend{T_0} + w\parend{T_1} - w\parend{T_2} - w\parend{T_3} - w\parend{T_4}\\
			&=\; w\parend{T_0} + w\parend{T_1} - \parend{w\parend{T} - w\parend{T_0} - w\parend{T_1}} 
							&& w\parend{T_0} + ... + w\parend{T_4} = w\parend{T} \\
			&=\; 2\parend{w\parend{T_0} + w\parend{T_1}} - w\parend{T} \\
			&<\; 0			 && w\parend{T_0} + w\parend{T_1} < \half w\parend{T}
		\end{align*}
	\end{case}
	
	\begin{case}[Either $b_1$ or $b_2$ correspond to a less-than test]
		If $b_1$ or $b_2$ correspond to a less-than-test, let $T'$ be rooted at (either) said test, then place the remaining tests from the original trees into the appropriate branch, performing them in the same order as in $T$. We may assume wlog that we choose $b_1$. Since we are not introducing a new tests, no fracture can occur in this case. We break into cases depending on which test $v_i$ (which we assume is a less-than test) is rooted at $T'$.
		
		Because all other $v_i$ lie on the same branch below $v_0$, $k_0$ is either the least or greatest key. So $k_0$ can not be $b_1$, (unless $b_1$ corresponds to an equal-to test on $k_0$, but we assume $b_1$ corresponds to a less-than-test). So $v_0$ can't be the root of $T'$
		
		If $v_1$ is the root of $T'$, then $v_0$ is right below $v_1$, moving $T_0$ down by 1 at depth 2. $v_0$ will be in whatever outcome corresponds to the side branch of $v_1$ in the original tree, which implies that the main branch of $v_0$ in the new tree is $T_1$, staying at the same depth 2. That means that $v_2$ and $v_3$ are on the same branch of $v_1$, implying $T_2$, $T_3$, and $T_4$ all go up by 1. This gives a change in cost of $\Delta = w\parend{T_0} + w\parend{T_1} - w\parend{T_2} - w\parend{T_3} - w\parend{T_4} < 0$.
		
		If $v_2$ is the root of $T'$, then $T_2$ and $v_3$ will be on opposite branches. If $v_1$ and $v_3$ are on the same branch of the root, then $v_0$ is alone on one side, with $T_2$ being it's main branch. $T_0$ goes down by 1, $T_1$ doesn't move, and $T_2$, $T_3$, and $T_4$ go go up by 1, giving $\Delta = w\parend{T_0} - w\parend{T_2} - w\parend{T_3} - w\parend{T_4} < 0$. If $v_1$ and $v_3$ are on opposite sides, then $v_3$ is alone and $T_2$ is below $v_1$.  $T_0$ and $T_1$ goes down by 1, $T_2$ doesn't move, and $T_3$ and $T_4$ both go up by 2. The change in cost is
		
		\begin{align*}
			\cost(T') - \cost(T) \;&=\; w\parend{T_0} + w\parend{T_1} - 2\parend{w\parend{T_3} + w\parend{T_4}}	 \\
			&=\; w\parend{T_0} + w\parend{T_1} - 2\parend{w\parend{T} - w\parend{T_0} - w\parend{T_1} - w\parend{T_2}}  \\
			&=\; 3w\parend{T_0} + 3w\parend{T_1} + 2w\parend{T_2} - 2w\parend{T} \\
			&\le\; 4w\parend{T_0} + 4w\parend{T_1} - 2w\parend{T} && w\parend{T_2} \le w\parend{T_1} \le w\parend{T_0} \\
			&<\; 0  &&w\parend{T_0} + w\parend{T_1} < \half w\parend{T}
		\end{align*}
		
		If $v_3$ is the root of $T_3$, then $v_0$ is below $v_3$, moving $T_0$ down by 1. Either $T_1$ and $T_2$ stay at the same depth, or $T_1$ goes down by 1 and $T_2$ goes up by 1. $T_3$ and $T_4$ goes up by at least 1. This gives a change in cost $\Delta \le w\parend{T_0} + w\parend{T_1} - w\parend{T_2} - w\parend{T_3} - w\parend{T_4} < 0$. \qedhere
	\end{case}
	
%	Instead of restating the details, we'll note the differences that justify construction validity. Let \(T_i\) refer the side-branch of node \(v_i\). In Case 1 of Anderson et. al's proof, it is necessary to ensure that trees \(T_3\) and \(T_4\) are the only ones that can potentially fracture in their construction. The justification that ensures \(T_1\) and \(T_2\) don't split applies to \(T_0\). In Case 2, \(b_0\) can technically be a potential ``dividing cut'' if it corresponds to a \(\angled{<}\) test, which could change the analysis of Anderson et al.'s proof. But this can't happen, as \(v_0\) being the first comparison in the tree implies that \(b_0\) is either least or greatest key, and thus can't be a middle cut-point.
	
%	Applying the same analysis as in Anderson et al., we get a new tree where the worst potential changes in cost are \(\Delta \le 2\parend{w\parend{T_0} + w\parend{T_1}} - w\parend{T}\) and \(\Delta \le 3w\parend{T_0} + 5w\parend{T_1} - w\parend{T}\). Since \(T_0\) and \(T_1\) correspond to the side weights of \(v_0\) and \(v_1\), we have that \(w\parend{T_1} \le w\parend{T_0}\) and \(w\parend{T_0} + w\parend{T_1} < \frac{1}{2} w\parend{T}\). This implies that \(\Delta < 0\), meaning \(T\) is not optimal.
	
%	\setcounter{case}{0}

\end{proof}

%%%%%%%%%%%%%%%%%%%%

\begin{corollary}\label{corollary: Second Heaviest Key Less-Than Thresshold}
	Consider an instance of $n \ge 2$ keys with total weight $W$. Let $\alpha$ and $\beta$ be the first and second heaviest weights respectively, where
	
	\[\alpha + \beta < \tfrac{1}{2}W.\]
	
	Then there exists no optimal tree that's rooted at two consecutive equal-to tests.
\end{corollary}

\begin{proof}
If any tree is rooted at two consecutive equal-to tests with the above conditions, then by Lemma~\ref{lemma: RMLK property}, 
$sw\parend{v_0} = \alpha$ and $sw\parend{v_1} = \beta$, which contradicts~(\ref{equation: side weights thresholds}).
\end{proof}

%%%%%%%%%%%%%%%%%%%%

\begin{corollary}\label{corollary: First Side Weight Threshold}
	For any optimal tree $T$ rooted at $v_0$ for some instance of keys, $sw\parend{v_0} \ge \onefourth w\parend{T}$.
\end{corollary}

\begin{proof}
	This follows immediately from Lemma~\ref{lemma: side-weight monoticity} and Theorem~\ref{theorem: Side Weights Thresholds}.
\end{proof}

%%%%%%%%%%%%%%%%%%%%%%%%%%%%%%%%%%%%%%%%%%%%%%%%%%%%%%%%%%%%%%%%%%%%%%%%%%%%%%
%%%%%%%%%%%%%%%%%%%%%%%%%%%%%%%%%%%%%%%%%%%%%%%%%%%%%%%%%%%%%%%%%%%%%%%%%%%%%%

\subsection{Tightness of Thresholds}
\label{sec: tightness of thresholds}
%%%%%%%%%%%%%%%%%%%%%%%%%%%%%%%%%%%%%%%%%%%%%%%%%%%%%%%%%%%%%%%%%%%%%%%%%%%%%%%
%%%%%%%%%%%%%%%%%%%%%%%%%%%%%%%%%%%%%%%%%%%%%%%%%%%%%%%%%%%%%%%%%%%%%%%%%%%%%%%
%
%\subsection{Tightness of Thresholds}
%\label{sec: tightness of thresholds}
%\input{03-4_tightness_of_thresholds.tex}
%
%%%%%%%%%%%%%%%%%%%%%%%%%%%%%%%%%%%%%%%%%%%%%%%%%%%%%%%%%%%%%%%%%%%%%%%%%%%%%%%
%%%%%%%%%%%%%%%%%%%%%%%%%%%%%%%%%%%%%%%%%%%%%%%%%%%%%%%%%%%%%%%%%%%%%%%%%%%%%%%

As before, let $\alpha$ and $\beta$ be the heaviest weights, with $\alpha\ge\beta$. 
We now consider the regions in the $(\alpha,\beta)$-plane 
corresponding to guaranteed equal-to tests, no equal-to tests, or to no two-consecutive equal-to tests.
Some crucial properties have already been established in~\cite{atalig_chrobak_tight_threshold_23} and~\ref{theorem: Second Heaviest Key Equality Threshold}.
These bounds, together with the observations given below determine the regions illustrated in Figure~\ref{fig: thresholds graph}.
Notably, in this figure, all known bounds are straight lines, some of which form non-convex polygons. 

\begin{figure}[t]
	\centering
	\includegraphics{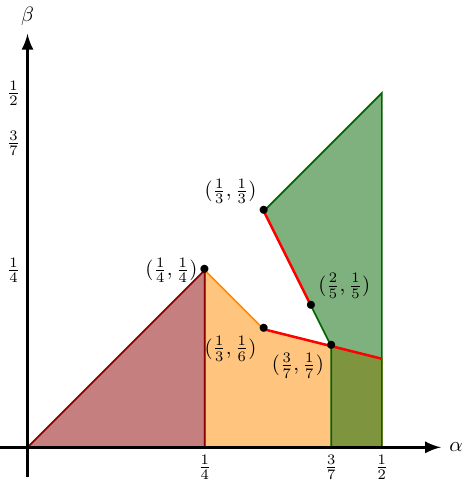}
	\caption{Plot of the thresholds on the first and second heaviest keys $\alpha$ and $\beta$, where the total weight is $1$. 
	The red region corresponds to no equal-to test, 
	orange correponds to no two-consecutive equal-to tests, and green corresponds to guaranteed equal-to test. The bright red lines indicate tight borders.}
	\label{fig: thresholds graph}
\end{figure}

Note that in addition to the bound given by Corollary~\ref{corollary: Second Heaviest Key Less-Than Thresshold}, 
we can trivially use $\lowlambda = \onefourth$ to establish another bound. We prove that this one is tight.

%%%%%%%%%%%%%%%%%%%%%%%%%%%

\begin{observation}\label{obs: other Second Heaviest Key Less-Than Threshold}
	Consider an instance of $n\ge 2$ keys with total weight $W$. Let $\alpha$ and $\beta$ be the first and second heaviest weights respectively, where 
	$\alpha + 4\beta < W$.
	Then there exists no optimal tree rooted at two consecutive equal-to tests.
\end{observation}

Note that inequality $\alpha+4\beta < W$ implies that we in fact have $n\ge 5$.

\begin{proof}
If an optimal tree's root is an equal-to test, then it is on the $\alpha$-weight key, with the right branch having weight $W-\alpha$ 
and heaviest key having weight $\beta$. By the lemma's assumption, 
$\beta < \onefourth\parend{W - \alpha}$ and thus the next comparison can not be an equal-to test, by Theorem~\ref{theorem: lambda-}.
\end{proof}

%%%%%%%%%%%%%%%%%%%%%%%%%%%

\begin{observation}\label{obs: Tightness on Second Heaviest Key}
	The bound given by Theorem~\ref{theorem: Second Heaviest Key Equality Threshold} is tight for $\alpha < \twofifths$.
\end{observation}

\begin{proof}
	Let $\beta \le \alpha$, $2\alpha + \beta  = 1$, and $\epsilon > 0$. 
	Consider the instance with four keys and weights given in the table below:
	
	\newlength{\mycolwidth}
	{
	\medskip
	\begin{center}
	\setlength{\mycolwidth}{0.42in}\renewcommand{\arraystretch}{1.1}
	\begin{tabular}{| m{\mycolwidth} | m{\mycolwidth} | m{\mycolwidth} | m{\mycolwidth} |} \hline
		{1}  & {2}  & {3}  & {4}  \\ \hline 
		{$\half\alpha$} & {$\alpha$} & {$\beta-\epsilon$}  & {$\half\alpha+\epsilon$}  \\ \hline
	\end{tabular}
	\end{center}
	\medskip
	}
	
	Note that for $\alpha < \twofifths $ we have $\half \alpha < \beta$, so if $\epsilon$ is sufficiently small then
	$\alpha$ and $\beta - \epsilon$ are the two heaviest weights. 
	The best tree starting with an equal-to test is uniquely given by performing $\equalstest{2}$ then $\equalstest{3}$ 
	with cost $C^= = 4\alpha + 2\beta  + \epsilon$, 
	while the best tree starting with a less-than test is given by performing $\lessthantest{3}$ with cost $C^< = 4\alpha + 2\beta$. 
	Then $C^< < C^=$ and
	the two heaviest weights satisfy $2\alpha  + (\beta -\epsilon)  < 1$.
	 (In Figure~\ref{fig: thresholds graph} this corresponds to the vector $\parend{0 ,-\epsilon}$ pointing down from the line $2\alpha + \beta = 1$.)
	 This holds for $\epsilon$ arbitrarily small.
\end{proof}

%%%%%%%%%%%%%%%%%%%%%%%%%%%

\begin{observation}\label{obs: Tightness on Two Consecutive Equals Tests}
	The bound given by Observation~\ref{obs: other Second Heaviest Key Less-Than Threshold} is tight for $\alpha \ge \onethird$.
\end{observation}

\begin{proof}
	Let $\beta \le \alpha$, $\alpha + 4\beta = 1$, and $\alpha\ge \onethird$.
	Consider the instance with eight keys  whose weights are given in the table below

	{
	\medskip
	\begin{center}
	\setlength{\mycolwidth}{0.1in}
	\begin{tabular}{| m{\mycolwidth} | m{\mycolwidth} | m{\mycolwidth} | m{\mycolwidth} | m{\mycolwidth} | m{\mycolwidth} | m{\mycolwidth} | m{\mycolwidth} |} \hline
		{1}  & {2}  & {3}  & {4} & 5 & 6 & 7 & 8 \\ \hline
		{$\alpha$} & {$0$} & {$\beta$}  & {$0$} & $\beta$ & $\beta$ & $0$ & $\beta$  \\ \hline
	\end{tabular}
	\end{center}
	\medskip
	}
	
	Note that $\alpha \ge \onethird$ implies $\beta \le \onesixth$. By case analysis, one can show that the optimal cost starting at an equal-to test is $C^= = \alpha + 14 \beta$, 
	which can be achieved by doing only equal-to tests. 
	(Alternatively, start with $\equalstest{1}$, then do $\lessthantest{5}$, followed by equal-to tests.)
	Likewise, the optimal cost starting at a less-than test is $C^< = 2\alpha + 12\beta$, which can be achieved 
	using $\lessthantest{4}$ or $\lessthantest{5}$ or $\lessthantest{6}$ as the root, followed by only equal-to tests in each branch of the root. We have $C^= - C^< = 2\beta - \alpha$, and by $\alpha \ge \onethird$ and $\beta \le \onesixth$ we have $C^= \le C^<$.
\end{proof}

%%%%%%%%%%%%%%%%%%%%%%%%%%%%%%%%%%%%%%%%%%%%%%%%%%%%%%%%%%%%%%%%%%%%%%%%%%%%%%
%%%%%%%%%%%%%%%%%%%%%%%%%%%%%%%%%%%%%%%%%%%%%%%%%%%%%%%%%%%%%%%%%%%%%%%%%%%%%%

%%%%%%%%%%%%%%%%%%%%%%%%%%%%%%%%%%%%%%%%%%%%%%%%%%%%%%%%%%%%%%%%%%%%%%%%%%%%%%
%%%%%%%%%%%%%%%%%%%%%%%%%%%%%%%%%%%%%%%%%%%%%%%%%%%%%%%%%%%%%%%%%%%%%%%%%%%%%%
%%%%%%%%%%%%%%%%%%%%%%%%%%%%%%%%%%%%%%%%%%%%%%%%%%%%%%%%%%%%%%%%%%%%%%%%%%%%%%

\section{Structure Results for Geometric Sequences}
\label{sec: structure results for geometric sequences}

%%%%%%%%%%%%%%%%%%%%%%%%%%%%%%%%%%%%%%%%%%%%%%%%%%%%%%%%%%%%%%%%%%%%%%%%%%%%%%
%%%%%%%%%%%%%%%%%%%%%%%%%%%%%%%%%%%%%%%%%%%%%%%%%%%%%%%%%%%%%%%%%%%%%%%%%%%%%%

%\section{Structure Results for Geometric Sequences}
%\label{sec: structure results for geometric sequences}
%\input{04_structure_results_for_geometric_sequences.tex}

%%%%%%%%%%%%%%%%%%%%%%%%%%%%%%%%%%%%%%%%%%%%%%%%%%%%%%%%%%%%%%%%%%%%%%%%%%%%%%
%%%%%%%%%%%%%%%%%%%%%%%%%%%%%%%%%%%%%%%%%%%%%%%%%%%%%%%%%%%%%%%%%%%%%%%%%%%%%%

As we show later in Section~\ref{sec: non-monotonicity of dynamic-programming minimizers}, it appears that the most difficult 
instances are ones whose weights form a geometric sequence. For 
that reason, understanding the properties of such instances and finding a special way of handling them may yield an algorithm 
that beats $O\parend{n^4}$.

Think of the dynamic programming algorithm as a memoized recursive algorithm.
What instances will be generated? Rescale the weights so that the heaviest
weight is $1$. If the weights decrease faster than $(1-\highlambda)^i$ 
(that is $\gamma < 1-\highlambda$) then, by the bound in~\cite{atalig_chrobak_tight_threshold_23},
the optimum choice at each step is an equal-to test, because after removing
the heaviest key, the weight of the new heaviest key is more than
$\highlambda$ times the new total.  That gives a trivial algorithm with complexity $O(n\log n)$.

On the other hand, suppose that the weights decrease slower than $\large(\threefourths\large)^i$.
(That is, $\gamma > \threefourths = 1 - \lowlambda$.) Then
the situation is not quite symmetric. At the first step we know for certain we will
make a less-than test, but since the weights will get partitioned between the
two intervals, the maximum key weights in these intervals could be large
compared to their total weights, so for these sub-instances equal-to tests may be
optimal. 

Assume now that the weights of an instance form the (infinite) geometric
sequence $(\gamma^i)_{i\ge 0}$. Then the total weight of the instance is $1/(1-\gamma)$
and the heaviest weight $1$ is a $1-\gamma$ fraction of the total.
Define $\lowgamma$ be the largest value such that for all $\gamma \in (0,\lowgamma)$ and
for each geometric instance with ratio $\gamma$ there is an optimal tree starting with the equal-to test.
By the bound on $\highlambda$ in~\cite{atalig_chrobak_tight_threshold_23}, we know that $\lowgamma\ge 1-\highlambda$ = $\frac{4}{7}$.
Below we give a tight bound.
		
%%%%%%%%%%

\begin{theorem}\label{theorem: lowgamma}
	$\lowgamma = \phi = \parend{\sqrt{5}-1}/2 \approx 0.618$.
\end{theorem}

\begin{proof} $(\le)$
	To show that $\lowgamma \le \phi$,
	consider the gemeotric instance with weights $(\gamma^i)_i$, and with the corresponding keys
	ordered from left to right in decreasing order of weight. 
	Let $C^=$ and $C^<$ be the optimal costs of trees starting with the equal-to test
	and less-than test, respectively. If starting with the equal-to test is optimal
	for ratio $\gamma$, then the optimal tree will have only equal-to tests, 
	and its cost is 
	\begin{equation*}
		C^= \;=\; \sum_i (i+1) \gamma^i \;=\; \frac{1}{ (1-\gamma)^2 }
	\end{equation*}
	The tree for $C^<$ that splits the instance between $\gamma$ and $\gamma^2$
	will apply the same split recursively, so computing its cost we get 
	
	\begin{equation*}
		C^< \;\le\; \sum_i (i+2)(\gamma^{2i} + \gamma^{2i+1}) 
		\;=\; 	(\gamma+1) \sum_i (i+2) \gamma^{2i}
		\;=\; \frac{1}{(1-\gamma)^2} \cdot \frac{2-\gamma^2}{1+\gamma}
	\end{equation*}
	Equating these values, gives $2-\gamma^2 = 1+\gamma$, so $\gamma = \phi = 0.681 ...$.
	If $\gamma > \phi$ then $2-\gamma^2 < 1+\gamma$, in which case $C^< < C^=$.
	So for $\gamma > \phi$, this instance does not have an optimal tree starting with the equal-to test.

$(\ge)$ We now show that $\lowgamma \ge \phi$.
	Consider an arbitrary geometric instance with ratio $\gamma$, for some $\gamma < \phi$.
	We will show that any such instance has a search tree that uses only
	equal-to tests. We will in fact show something stronger, proving that this
	property holds even if we allow more general instances and queries:
	\begin{itemize}
		\item The key weights in the instance can be any subset of the geometric
		sequence $(\gamma^i)_i$ (not necessarily the whole sequence).
		\item The tree can use any pre-specified set of arbitrary queries, as long as it includes
			the equals-to tests to all keys.
			We refer to queries in this set as \emph{allowed}.
			Each allowed query can be represeted as $\angled{\in X}$, where $X$ 
			is the set of keys that satisfies this query. 
			The equal-to test $\equalstest{k}$ is a special case when $X = \braced{k}$.
	\end{itemize} 
	For convenience, we identify keys by their weights, so, for example,
	 $\equalstest{1}$ refers to an equal-to test on the key of weight $1$.
	
	It is now sufficient to prove the following claim: Let $T$ be a search tree using allowed queries for an
	instance of $n$ keys whose weights are a subset of the sequence $(\gamma^i)_i$.
	Then $T$ can be converted, without increasing cost, into a search tree $T'$ that 
	uses only equals-to tests and its root is an equals-to test to the heaviest key.
 	
	The claim is trivial for $n=1,2,3$. So assume that $n\ge 4$
	and that the claim holds for trees for instances that have fewer than $n$ keys.
	Consider an instance with $n$ keys. Without loss of generality, we can assume that $1$ is in this instance. 
	Let $T$ be a tree for this instance, and let $L$ and $R$ be its left (yes) and right (no) subtrees.
	By induction, $L$ and $R$ have only equals-to tests, with their heaviest key at their
	respective roots.

	If $L$ or $R$ is a leaf, then we can assume that
	the root of $T$ is an equals-to test. If this test is $\equalstest{1}$,
	then we are done, that is $T' = T$. Else, this test is $\equalstest{\gamma^a}$ for some $a\ge 1$,
	and $R$'s root is $\equalstest{1}$. We can then obtain $T'$ by swapping these two
	equals-to nodes.

%%%%%%%%%%%%%%%%%%%
	
	\begin{figure}[t]
		\begin{center}
			\begin{subfigure}[b]{0.44\textwidth}
				\includegraphics[valign=m, scale=0.9]{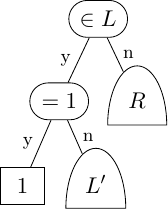}
				\quad$\boldsymbol{\Longrightarrow}$\quad
				\includegraphics[valign=m, scale=0.9]{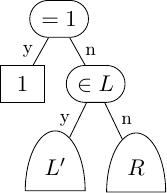}
				\caption{If $R$ doesn't contain $\gamma$.} \label{fig: lowgamma figa}
			\end{subfigure}
			\begin{subfigure}[b]{0.54\textwidth}
				\includegraphics[valign=m, scale=0.9]{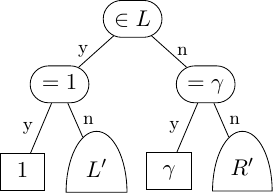}
				\quad$\boldsymbol{\Longrightarrow}$\quad
				\includegraphics[valign=m, scale=0.9]{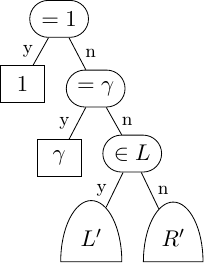}
				\caption{If $R$ contains $\gamma$.} \label{fig: lowgamma figb}
			\end{subfigure}
		\end{center}
		\caption{Illustration of the lower bound for $\lowgamma$.}
		\label{fig: lowgamma}
	\end{figure}

%%%%%%%%%%%%%%%%%%%	

	So from now on we can assume that each of $L$ and $R$ has at least two leaves.
	The root of $T$ is a query $\angled{\in L}$. (Slightly abusing notation, we use $L$ both
	for a subtree and its set of keys.)
	We can assume that $1$ is in $L$, for otherwise in the rest of the proof
	we can swap the roles of $L$ and $R$.
	Then, by the inductive assumption, the root of $L$ is $\equalstest{1}$.
	Let $L'$ be the no-subtree of the node $\equalstest{1}$.
	If $R$ does not contain $\gamma$ then its total weight is 
	at most $\gamma^2/(1-\gamma) < \phi^2/(1-\phi) = 1$, in which case we can
	rotate $\equalstest{1}$ upwards, obtaining tree $T'$ with smaller cost (see Figure~\ref{fig: lowgamma figa}).
	
	If $R$ contains $\gamma$ then, by the inductive assumption, its root
	is $\equalstest{\gamma}$. Let $R'$ be the no-subtree of $\equalstest{\gamma}$. 
	Then the largest key weight in $L'\cup R'$ is at most $\gamma^2$,
	so the total weight of $L'\cup R'$ is less than $1$ (as in the argument above).
	Then we can rotate upwards both tests $\equalstest{1}$ and $\equalstest{\gamma}$, 
	as in Figure~\ref{fig: lowgamma figb}, obtaining our new tree $T'$ with smaller cost.
\end{proof}

We can similarly define the other threshold for geometric sequences:
Let $\highgamma$ be the smallest value such that for all $\gamma \in (\highgamma,1)$
and each geometric instance with ratio $\gamma$ the optimal tree for
this instance does not have the equal-to test at the root.
By the bound on $\lowlambda$, we know that $\highgamma \le 1-\lowlambda = \threefourths$.
Some improvements to this bound are probably not difficult, but
establishing a tight bound seems more challenging.

%%%%%%%%%%%%%%%%%%%%%%%%%%%%%%%%%%%%%%%%%%%%%%%%%%%%%%%%%%%%%%%%%%%%%%%%%%%%%%
%%%%%%%%%%%%%%%%%%%%%%%%%%%%%%%%%%%%%%%%%%%%%%%%%%%%%%%%%%%%%%%%%%%%%%%%%%%%%%
%%%%%%%%%%%%%%%%%%%%%%%%%%%%%%%%%%%%%%%%%%%%%%%%%%%%%%%%%%%%%%%%%%%%%%%%%%%%%%

\section{Attempts at Dynamic Programming Speedup}
\label{sec: non-monotonicity of dynamic-programming minimizers}

%%%%%%%%%%%%%%%%%%%%%%%%%%%%%%%%%%%%%%%%%%%%%%%%%%%%%%%%%%%%%%%%%%%%%%%%%%%%%%
%%%%%%%%%%%%%%%%%%%%%%%%%%%%%%%%%%%%%%%%%%%%%%%%%%%%%%%%%%%%%%%%%%%%%%%%%%%%%%

%\section{Non-Monotonicity of Dynamic-Programming Minimizers}
%\label{sec: non-monotonicity of dynamic programming minimizers}
%\input{05_non-monotonicity_of_dynamic-programming_minimizers.tex}

%%%%%%%%%%%%%%%%%%%%%%%%%%%%%%%%%%%%%%%%%%%%%%%%%%%%%%%%%%%%%%%%%%%%%%%%%%%%%%
%%%%%%%%%%%%%%%%%%%%%%%%%%%%%%%%%%%%%%%%%%%%%%%%%%%%%%%%%%%%%%%%%%%%%%%%%%%%%%

In the dynamic programming formulations of some optiomization problems,
the optimum cost function satisfies the so-called \emph{quadrangle inequality}.
In its different incarnations this property is sometimes referred to as
the \emph{Monge property} and is also closely related to the \emph{total monotonicity of matrices}
(see~\cite{bein_et_al_knuth-yao_total_monotonicity_09}).
The quadrangle inequality was used, in particular, to obtain
$O(n^2)$-time algorithms for computing optimal 3-way comparison search trees
\cite{Knuth1971,yao_efficient_dynamic_programming_80}.

Recall that in our dynamic programming formulation (Section~\ref{sec: preliminaries})
by $C^h_{i,j}$ we denote the minimum cost of the \twoWCST for $I^h_{i,j}$,
which is the sub-problem that includes the keys in the interval $[i,j]$
which are also among the $h$ lightest keys.
In this context,
the quadrangle inequality would say that for all $i \le i' \le j \le j'$ we have
\begin{eqnarray}
	C^h_{i,j} + C^h_{i',j'} &\le&
	C^h_{i,j'} + C^h_{i',j}.
	\label{eqn: quadrangle inequality}
\end{eqnarray}
The consequence of the quadrangle inequality is that the minimizers along the
diagonals of the dynamic programming matrix satisfy a certain monotonicity
property. Namely, for any $i,j,h$, let $L^h_{i,j}$ be the minimizer
for $S^h_{i,j}$, defined as the index $l$ such that $S^h_{i,j} = C^h_{i,l} + C^h_{l+1,j}$.
Then the quadrangle inequality implies that 
\begin{eqnarray}
	L^h_{i,j-1} \;\le\; L^h_{i,j} \;\le\; L^h_{i+1,j}.
	\label{eqn: minimizer monotonicity}
\end{eqnarray}
With this monotonicity property, we would only need to search for $L^h_{i,j}$ between
$L^h_{i,j-1}$ and $L^h_{i+1,j}$, thus in the diagonal $j-i$ we search disjoint
intervals. For some problems, conceivably, this property may be true even if the
quadrangle inequality is not, and it would be in itself sufficient to speed up dynamic programming.

In this section we show, however, that this approach does not apply. We show
that the quadrangle inequality~(\ref{eqn: quadrangle inequality}) is false,
and that the monotonicity property~(\ref{eqn: minimizer monotonicity}) is false
as well. We also give counter-examples to other similar properties that could
potentially help in designing a faster algorithm.

%%%%%%%%%%%%%%%%%%%%%%%%%%%%%%%%%%%%%%%%%%%%%%%%%%%%%%%%%%%%%%%%%%%%%%%%%%%%%%
%%%%%%%%%%%%%%%%%%%%%%%%%%%%%%%%%%%%%%%%%%%%%%%%%%%%%%%%%%%%%%%%%%%%%%%%%%%%%%

\subsection{Quadrangle Inequality and Minimizer Monotonicity}
\label{sec: Quadrangle Inequality and Minimizer Monotonicity}
%%%%%%%%%%%%%%%%%%%%%%%%%%%%%%%%%%%%%%%%%%%%%%%%%%%%%%%%%%%%%%%%%%%%%%%%%%%%%%%
%%%%%%%%%%%%%%%%%%%%%%%%%%%%%%%%%%%%%%%%%%%%%%%%%%%%%%%%%%%%%%%%%%%%%%%%%%%%%%%
%
%\subsection{Quadrangle Inequality and Minimizer Monotonicity}
%\label{sec: Quadrangle Inequality and Minimizer Monotonicity}
%\input{05-1_quadrangle_inequality_and_minimizer_monotonicity.tex}
%
%%%%%%%%%%%%%%%%%%%%%%%%%%%%%%%%%%%%%%%%%%%%%%%%%%%%%%%%%%%%%%%%%%%%%%%%%%%%%%%
%%%%%%%%%%%%%%%%%%%%%%%%%%%%%%%%%%%%%%%%%%%%%%%%%%%%%%%%%%%%%%%%%%%%%%%%%%%%%%%

\paragraph{Counter-examples to quadrangle inequality.}
Here is an example that the quadrangle inequality~(\ref{eqn: quadrangle inequality}) does not hold
in our case. We have three ikeys $1,2,3$, with weights $1,V,1$. Let $i=1, i' = 2 = j, j' = 3$.
We claim that
\begin{eqnarray*}
	C^3_{1,2} + C^3_{2,3} &>&
	C^3_{1,3} + C^3_{2,2}.
\end{eqnarray*}
Indeed, $C^3_{1,2} = V+1$, $C^3_{2,3} = V+1$, while
$C^3_{1,3} = V+4$ and $C^3_{2,2} = 0$.
Thus the quadrangle inequality~(\ref{eqn: quadrangle inequality}) is false.

This easily generalizes to larger instances. Let the keys be $1,2,...,n$, for $n = 2^{s+1}+1$.
The instance consists of a block of $2^s$ keys of weight $1$, followed by a key
of weight $V$, followed by another block of $2^s$ keys of weight $1$.
Let $j = 2^s+1$ be the key of weight $V$.

The optimum tree for an instance consisting of $2^s$ keys of weight $1$ is a perfectly
balanced tree and its cost is $s2^s$. So, if $V$ is large enough then
$C^n_{1,j} = C^n_{j,n} = V + (s+1)2^s$ and $C^n_{1,n} = V+ (s+2)2^{s+1}$.
Since $C^n_{j,j} = 0$, we then have
\begin{equation*}
	C^n_{1,j} + C^n_{j,n} \;=\; 2V + (s+1)2^{s+1} \;>\; C^n_{j,j} + C^n_{1,n},
\end{equation*}
as long as $V > 2^{s+1}$.

%%%%%%%%%%%%%

\paragraph{Counter-example to minimizer monotonicity.}
We now consider~(\ref{eqn: minimizer monotonicity}), showing that it's
also false. Consider six keys, with the following weights:

\smallskip
\begin{center}
\begin{tabular}{|c|c|c|c|c|c|}\hline
	1 & 2 & 3 & 4 & 5 & 6 \\ \hline
	1 & V & 1 & 1 & V & 1 \\ \hline
\end{tabular}
\end{center}
\smallskip

\noindent
where $V > 8$. We claim that the unique minimizers for $C^6_{1,6}$ and $C^6_{2,6}$ satisfy
\begin{eqnarray*}
	L^6_{1,6} \;=\; 3 \; > \; 2 \;=\; L^6_{2,6}.
\end{eqnarray*}
Indeed, we have $C^6_{1,3} + C^6_{4,6} = (V+4) + (V+4) = 2V+8$
(each cost is realized by a tree with the equal-to test for the heavy item in the root).
On the other hand, $C^6_{1,1} + C^6_{2,6} \ge 0 + 3V = 3V$, since in any
tree for $I^6_{2,6}$ one heavy item will be at depth at least $2$. Also,
$C^6_{1,2} + C^6_{3,6} \ge (V+1) + (V+8) = 2V+9$,
because in an optimal tree for $I^6_{3,6}$ the first 
test must be an equal-to test for the heavy item,
and the remaining items of weight $1$ will have
depths at least $2,3,3$. The remaining expressions
are symmetric. Thus proves that $L^6_{1,6} = 3$.

To show that $L^6_{2,6} = 2$, we first compute
$C^6_{2,2} + C^6_{3,6} = 0 + (V+8)  = V+8$.
On the other hand, for $l= 3,4$, we have
$C^6_{2,l} + C^6_{l,6} \ge 2V$, since each
$V$ will be in a tree of size at least $2$,
while $C^6_{2,5} + C^6_{6,6} \ge 3V+0 = 3V$.

%%%%%%%%%%%%%%%%%%%%%%%%%%%%%%%%%%%%%%%%%%%%%%%%%%%%%%%%%%%%%%%%%%%%%%%%%%%%%%
%%%%%%%%%%%%%%%%%%%%%%%%%%%%%%%%%%%%%%%%%%%%%%%%%%%%%%%%%%%%%%%%%%%%%%%%%%%%%%

\subsection{Monotonicity without Boundary Keys}
\label{sec: Monotonicity without Boundary Keys}
%%%%%%%%%%%%%%%%%%%%%%%%%%%%%%%%%%%%%%%%%%%%%%%%%%%%%%%%%%%%%%%%%%%%%%%%%%%%%%%
%%%%%%%%%%%%%%%%%%%%%%%%%%%%%%%%%%%%%%%%%%%%%%%%%%%%%%%%%%%%%%%%%%%%%%%%%%%%%%%
%
%\subsection{Monotonicity without Boundary Keys}
%\label{sec: Monotonicity without Boundary Keys}
%\input{05-2_monotonicity_without_boundary_keys.tex}
%
%%%%%%%%%%%%%%%%%%%%%%%%%%%%%%%%%%%%%%%%%%%%%%%%%%%%%%%%%%%%%%%%%%%%%%%%%%%%%%%
%%%%%%%%%%%%%%%%%%%%%%%%%%%%%%%%%%%%%%%%%%%%%%%%%%%%%%%%%%%%%%%%%%%%%%%%%%%%%%%

We now make the following observation. In formula~(\ref{eqn: formula for S}),
in the cases when $l=j$ or $l = j-1$, instead of making a split with an
less-than test, we can make \emph{equal-to} tests to $i$ or ${j}$. 
But we know that, without loss of generality, equal-to tests are made
only to heaviest items. If $i$ or $j$ is the heaviest key, then the
option of doing the equal-to test will be included in the minimum~(\ref{eqn: formula for C})
anyway. If none of them is a heaviest key, then the corresponding equal-to tests
can be ignored. As a result, in~(\ref{eqn: formula for S}) we can only
run index $l$ from $i+1$ to $j-2$ (and assume that $j\ge i+4$). 

Motivated by this, we now introduce a definition of $S^h_{i,j}$, slightly
different from the one in the previous section.
$S^h_{i,j}$ is the minimum cost of a binary search tree
with the following two properties:
(i) the root of the tree is an less-than-comparison (split node),
and (ii) both the left and the right subtree of the root has at
least two nodes. Note that this definition implies that 
$S^h_{i,j} = \infty$ whenever $|I^h_{i,j}|\le 3$.

We now give a modified recurrence for $C^h_{i,j}$. If 
$|I^h_{i,j}|\le 1$ then $C^0_{i,j} = 0$. 
Assume $h \ge 1$ and $|I^h_{i,j}|\ge  2$ (in particular, $i < j$).
If $a_h\notin I^h_{i,j}$, then $C^h_{i,j} = C^{h-1}_{i,j}$,
because in this case $I^h_{i,j} = I^{h-1}_{i,j}$.
So assume now that $a_h\in I^h_{i,j}$.  Then:
\begin{align}
	C^h_{i,j} \;&=\; w^h_{i,j} + 
	\min \braced{ C^{h-1}_{i,j} , S^h_{i,j} },  \quad\textrm{where}  
	\label{eqn: formula for C modified}
	\\
	S^h_{i,j} \;&=\; 
	\min_{l = i+1,...,j-2}\braced{ C^h_{i,l} + C^h_{l+1,j}}.
	\label{eqn: formula for S modified} 
\end{align}
The earlier counter-example does not apply to this new formulation, because
now we only need the minimizer monotonicity to hold between locations $l=i+1$ and $l=j-2$.
(The counter-example to the quadrangle inequality is technically correct, but
it does not seem to matter for our new definition.)

%%%%%%%%%%%%%%%

\paragraph{Counter-example.}   
This does not work either. Consider the example below.

\smallskip
\begin{center}
\begin{tabular}{|c|c|c|c|c|c|}\hline
	1 & 2 & 3 & 4 & 5 & 6 \\ \hline
	0 & 2 & 2 & 0 & 1 & 1 \\ \hline
\end{tabular}
\end{center}
\smallskip

We claim that the minimizer monotonicity is violated in this 
example.  Recall that $L^h_{i,j}$ is the minimizer
for $S^h_{i,j}$, namely the
index $l\in\braced{i+1,...,j-2}$ such that $S^h_{i,j} = C^h_{i,l} + C^h_{l+1,j}$.   
(Note that now we don't allow values $l=i$ and $l=j-1$.)
To get the amortization to work, it would be enough to have
\begin{eqnarray*}
	L^h_{i,j-1} \;\le\; L^h_{i,j} \;\le\; L^h_{i+1,j}.
\end{eqnarray*}
We show that this inequality is false. 
In this example we have $L_{1,6} = 4 > 3 = L_{2,6}$. 
(Below
we use $h = 6$ in all expressions, so we will omit it.)

To prove this, let's consider $L_{1,6}$ first.
For $l=2$ we have $C_{1,2} +C_{3,6} = 2+7 = 9$.  
For $l=3$ we have $C_{1,3} + C_{4,6} = 6 + 3 = 9$.
For $l=4$ we have $C_{1,4} + C_{5,6} = 6+2 = 8$.
Thus $L_{1,6} = 4$, as claimed.

Next, consider $L_{2,6}$.  
For $l=3$ we have $C_{2,3} + C_{3,6} = 4 + 3 = 7$.
For $l=4$ we have $C_{2,4} + C_{5,6} = 6 + 2 = 8$.
Thus $L_{2,6} = 3$, as claimed.

Note also that this example remains valid even
if we don't restrict $l$ to be between $i+1$ and $j-2$.
For $L_{1,6}$, we have that  $C_{1,1} +C_{2,6}$
and $C_{1,5}+C_{6,6}$ are at least $9$.
For $L_{2,6}$, we get that
$C_{2,2}+C_{3,6} = 7$ and $C_{2,5}+C_{6,6} = 9$, so
$l = 2$ is tied with $l=3$. But we can
break the tie to make $l=3$ by slightly increasing the weight 
of key $6$ to $1+\epsilon$, for some small $\epsilon > 0$.

%%%%%%%%%%%%%%%%%%%%%%%%%%%%%%%%%%%%%%%%%%%%%%%%%%%%%%%%%%%%%%%%%%%%%%%%%%%%%%
%%%%%%%%%%%%%%%%%%%%%%%%%%%%%%%%%%%%%%%%%%%%%%%%%%%%%%%%%%%%%%%%%%%%%%%%%%%%%%

\subsection{Monotonicity For Constant Weights}
\label{sec: Monotonicity For Constant Weights}
%%%%%%%%%%%%%%%%%%%%%%%%%%%%%%%%%%%%%%%%%%%%%%%%%%%%%%%%%%%%%%%%%%%%%%%%%%%%%%%
%%%%%%%%%%%%%%%%%%%%%%%%%%%%%%%%%%%%%%%%%%%%%%%%%%%%%%%%%%%%%%%%%%%%%%%%%%%%%%%
%
%\subsection{Monotonicity For Constant Weights}
%\label{sec: Monotonicity For Constant Weights}
%\input{05-3_monotonicity_for_constant_weights.tex}
%
%%%%%%%%%%%%%%%%%%%%%%%%%%%%%%%%%%%%%%%%%%%%%%%%%%%%%%%%%%%%%%%%%%%%%%%%%%%%%%%
%%%%%%%%%%%%%%%%%%%%%%%%%%%%%%%%%%%%%%%%%%%%%%%%%%%%%%%%%%%%%%%%%%%%%%%%%%%%%%%

So far, our counter examples have used relatively unbalanced weights. Suppose that the weights were bounded in the range $\brackd{1, R}$, where $R$ is constant relative to $n$. 
For sufficiently large $n$, most of the comparisons in the search tree would have to be less-than tests, except with a few equal-to tests near the bottom. 
A natural conjecture would be that these perturbations near the bottom would not affect global structural properties of an optimal tree, and that
the quadrangle inequality could only fail for some constant number of diagonals near the main diagonal. As it turns out, this is false.

\paragraph{Counter-example.} Consider an instance of $n$ keys with a repeating pattern of weights $\parend{1, 3, 1, 3, \ldots}$, going from left to right. 
We claim without proof that if $n$ is an even number between $2^{p+1} +2$ and $3\cdot2^p$ for $p\ge 2$, we have

\[C_{1, n-1}^n + C_{2, n}^{n} = C_{2, n-1}^n + C_{1, n}^n + 1.\]

Furthermore, we claim that  $L_{2, n}^n = L_{3, n}^n + 1$  for even $n$ in $\brackd{2^{p+1} +2, 5\cdot 2^{p-1}}$ and $L_{1, n}^n = L_{1, n-1}^n -1$ for even $n$ in $\brackd{5\cdot 2^{p-1} + 2, 3\cdot2^{p}}$.
We have computationally verified these claims for $p\le8$. Intuitively, these instances correspond to balanced trees, where sub-trees rooted at equal-to tests are at the bottom. As shown in Figure~\ref{fig: balanced tree example}, these bottom sub-trees are the only things that change when we remove the boundary keys. Therefore the difference in cost, no matter the size of the tree, comes down to comparing the costs of small sub-trees, where equal-to tests are optimal and break monotonicity.

We also remark that this counter-example is not unique. Weight patterns like $\parend{1, 2, 3, 1, 2, 3, \ldots }$ and $\parend{3, 1, 1, 3, 1, 1, \ldots}$ yield similar results. In fact, monotonicity can break even by simply generating random integers in $\brackd{1, R}$, albeit in a less regular pattern. Figure~\ref{fig: QI Tables} illustrates some examples. Despite that, the error in monotonicity or the quadrangle inequality appears to be consistently small and not scale with $n$, suggesting that instances with constant weights may have an ``approximate monotonicity'' property
that could still speed up the DP algorithm.

\begin{figure}[t]
	\centering
	\begin{subfigure}[b]{0.49\textwidth}
		\centering
		\includegraphics[scale=0.75]{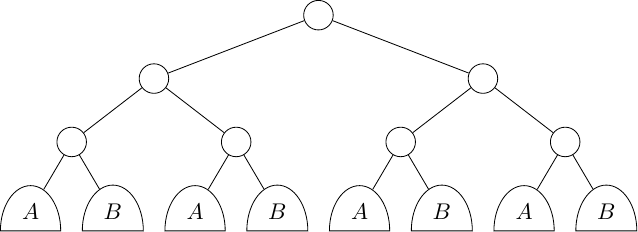}
		\caption{$I_{1, 24}^{24} $, $\text{cost} = 216$}\label{fig: balanced big tree}
	\end{subfigure}\hfill
	\begin{subfigure}[b]{0.49\textwidth}
		\centering
		\includegraphics[scale=0.75]{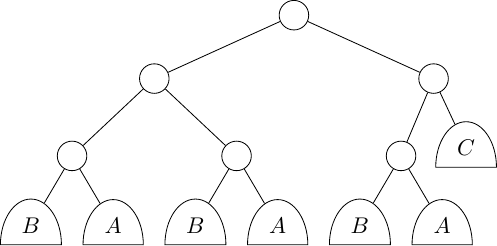}
		\caption{$I_{2, 23}^{24}$, $\text{cost} = 193$}\label{fig: balanced small tree}
	\end{subfigure}\\
	\begin{subfigure}[b]{0.49\textwidth}
		\centering
		\includegraphics[scale=0.75]{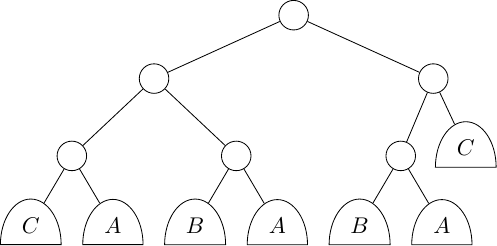}
		\caption{$I_{1, 23}^{24}$, $\text{cost} = 200$}\label{fig: balanced left tree}
	\end{subfigure}\hfill
	\begin{subfigure}[b]{0.49\textwidth}
		\centering
		\includegraphics[scale=0.75]{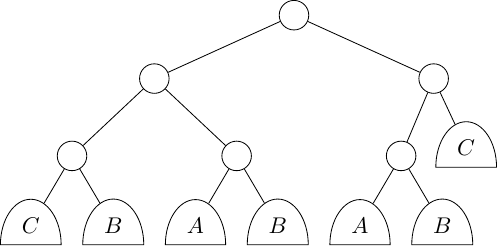}
		\caption{$I_{2, 24}^{24}$, $\text{cost} = 210$}\label{fig: balanced right tree}
	\end{subfigure}
	\caption{Optimal trees for the 24-key instance with weight pattern $\parend{1, 3, 1, 3, \ldots}$. $A$, $B$, and $C$ are sub-trees, where $A$ handles an instance with weights $\parend{1, 3, 1}$, $B$ handles $\parend{3, 1, 3}$, and $C$ handles $\parend{1, 3, 1, 3}$. Note that keys used in less-than tests, which are omitted from the figures, are different across these instances.} \label{fig: balanced tree example}
\end{figure}

\begin{figure}[t]
	\centering
	\begin{subfigure}[b]{0.32\textwidth}
		\includegraphics[width=\textwidth]{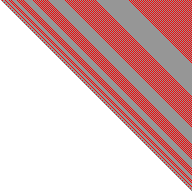}
		\caption{$\parend{1, 3, 1, 3, \ldots}$}\label{fig: QI Table 131313}
	\end{subfigure}~
	\begin{subfigure}[b]{0.32\textwidth}
		\includegraphics[width=\textwidth]{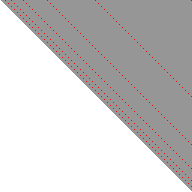}
		\caption{$\parend{3, 1, 2, 3, 1, 2, \ldots}$}\label{fig: QI Table 312312}
	\end{subfigure}~
	\begin{subfigure}[b]{0.32\textwidth}
		\includegraphics[width=\textwidth]{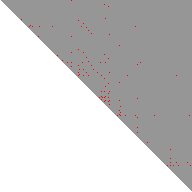}
		\caption{Random weights in $\braced{1, 2, 3}$}\label{fig: QI Table Random}
	\end{subfigure}~
	\caption{``QI tables'' for various weight patterns. The pixel $i$ rows down and $j$ columns to the right corresponds to $I_{i, j}^n$, where $n = 192$. We color $I_{i, j}^n$ gray if $C_{i+1, j}^n + C_{i, j-1}^n \le C_{i+1, j-1}^n + C_{i, j}^n$ or red if otherwise.} \label{fig: QI Tables}
\end{figure}

%%%%%%%%%%%%%%%%%%%%%%%%%%%%%%%%%%%%%%%%%%%%%%%%%%%%%%%%%%%%%%%%%%%%%%%%%%%%%%
%%%%%%%%%%%%%%%%%%%%%%%%%%%%%%%%%%%%%%%%%%%%%%%%%%%%%%%%%%%%%%%%%%%%%%%%%%%%%%

\subsection{Monotonicity Along Diagonals}
\label{sec: Monotonicity Along Diagonals}
%%%%%%%%%%%%%%%%%%%%%%%%%%%%%%%%%%%%%%%%%%%%%%%%%%%%%%%%%%%%%%%%%%%%%%%%%%%%%%%
%%%%%%%%%%%%%%%%%%%%%%%%%%%%%%%%%%%%%%%%%%%%%%%%%%%%%%%%%%%%%%%%%%%%%%%%%%%%%%%
%
%\subsection{Monotonicity Along Diagonals}
%\label{sec: Monotonicity Along Diagonals}
%\input{05-4_monotonicity_along_diagonals.tex}
%
%%%%%%%%%%%%%%%%%%%%%%%%%%%%%%%%%%%%%%%%%%%%%%%%%%%%%%%%%%%%%%%%%%%%%%%%%%%%%%%
%%%%%%%%%%%%%%%%%%%%%%%%%%%%%%%%%%%%%%%%%%%%%%%%%%%%%%%%%%%%%%%%%%%%%%%%%%%%%%%

Earlier we disproved the minimizer monotonicity property stating that
$L^h_{i,j-1} \le L^h_{i,j} \;\le\; L^h_{i+1,j}$. In this property, the minimizers
along the diagonal $d = j-i$ are sandwiched between the minimizers for the
previous diagonals. A weaker condition would be to require that 
that minimizers increase along  any given diagonal:

\begin{equation*}
	L^h_{i,j} \;\le\; L^h_{i+1,j+1}
\end{equation*}
We remark that this property does not necessarily imply a faster algorithm. 

\paragraph{Counter-example.} We now show that this property also fails. Consider the
example below. (This is a minor modification of the example in the previous section.)

\smallskip
\begin{center}
\begin{tabular}{|c|c|c|c|c|c|c|}\hline
	1 & 2 & 3 & 4 & 5 & 6 & 7 \\ \hline
	$\epsilon$ & 2 & 2 & 0 & 1 & 1 & 0 \\ \hline
\end{tabular}
\end{center}
\smallskip

Here, $\epsilon>0$ is very small. The same argument as in the previous section gives
us that $L_{1,6} = 4$. (The value of $C_{1,4} + C_{5,6}$ will be $8+O(\epsilon)$ while
for other splits the values are at least $9$.)
s
We now claim that $L_{2,7}\le 3$. We consider the split values $l=2,3,4,5,6$.
For $l=2$, $C_{2,2} + C_{3,7}\ge 11$.
For $l=3$, $C_{2,3} + C_{4,7}\le 7$.
For $l=4$, $C_{2,4} + C_{5,7}\ge 9$.
For $l =5$, $C_{2,5} + C_{6,7}\ge 10$.
For $l=6$, $C_{2,6} + C_{7,7}\ge 13$.
So $L_{2,7} = 3$.

Therefore we have $L_{1,6} = 4 > 3 = L_{2,7}$, showing that the
minimizers are not monotone along the diagonals of the dynamic programming table.

%%%%%%%%%%%%%%

\paragraph{Counter-example.} 
In the above example, the \emph{optimum tree} would actually start with an equal-to test.
To refine it, we take a step forward and present 
an instance where a less-than test in the root is necessary, while the minimizers still
do not meet the condition of monotonicity along the diagonals. 
The assigned weights are as follows.

\smallskip
\begin{center}
\begin{tabular}{|c|c|c|c|c|c|c|c|c|c|c|c|}\hline
	1 & 2 & 3 & 4 & 5 & 6 & 7 & 8 & 9 & 10 & 11 & 12 \\ \hline
	12 & 10 & 3 & 9 & 8 & 2 & 6 & 7 & 5 & 1 & 11 & 13 \\ \hline
\end{tabular}
\end{center}
\smallskip

Now we examine the \twoWCST's of $I_{1,11}^{12}$ and $I_{2,12}^{12}$. It can be simply verified 
that in both of these instances, the weight of the heaviest item is less than $\frac14$ of the 
total weight of the items. Thus, following Theorem~\ref{theorem: lambda-} (by Anderson et al.~\cite{Anderson2002}),
we observe that every \twoWCST should starts with a less-than test. Additionally, this example 
has another interesting property: For each $l \in \{2, 3, \dots, 10\}$, there exists a \twoWCST 
for $I_{1, l}^{12}$ and $I_{l + 1, 12}^{12}$ that initiates with less-than test. 
The same thing holds for $I_{2,12}^{12}$, which makes this set of items quite balanced. 
However, we have $L_{1,11}^{12} = 6 > 5 = L_{2,12}^{12}$ and these minimizers are unique.
This leads to a counter-example where the monotonicity of minimizers along the diagonals 
is violated (this result can be verified computationally). 
Figure~\ref{fig: counterexample monotonicity along diagonals} illustrates \twoWCST's of $I_{1,11}^{12}$
and $I_{2,12}^{12}$.

\begin{figure}[t]
	\centering
	\begin{subfigure}[b]{0.49\textwidth}
		\centering
		\includegraphics[scale=0.8, valign=m]{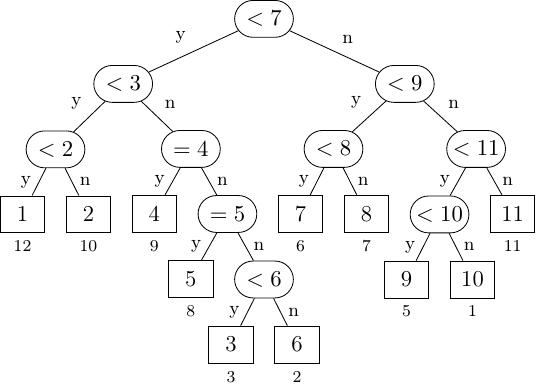}
		\caption{\twoWCST of $I_{1, 11}^{12}$, $L_{1,11}^{12} = 6$}
	\end{subfigure}
	\begin{subfigure}[b]{0.49\textwidth}
		\centering
		\includegraphics[scale=0.8, valign=m]{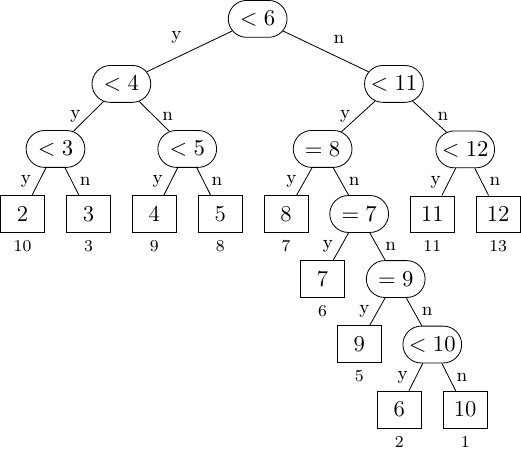}
		\caption{\twoWCST of $I_{2, 12}^{12}$, $L_{2,12}^{12} = 5$}
	\end{subfigure}
	\caption{Illustration of (a) $I_{1,11}^{12}$ and (b) $I_{2,12}^{12}$ in Section~\ref{sec: Monotonicity Along Diagonals}.} \label{fig: counterexample monotonicity along diagonals}
\end{figure}

%%%%%%%%%%%%%%%%%%%%%%%%%%%%%%%%%%%%%%%%%%%%%%%%%%%%%%%%%%%%%%%%%%%%%%%%%%%%%%
%%%%%%%%%%%%%%%%%%%%%%%%%%%%%%%%%%%%%%%%%%%%%%%%%%%%%%%%%%%%%%%%%%%%%%%%%%%%%%

\subsection{Marginal Advantage of Equal-to Tests}
\label{sec: Marginal Advantage of Equality Tests}
%%%%%%%%%%%%%%%%%%%%%%%%%%%%%%%%%%%%%%%%%%%%%%%%%%%%%%%%%%%%%%%%%%%%%%%%%%%%%%%
%%%%%%%%%%%%%%%%%%%%%%%%%%%%%%%%%%%%%%%%%%%%%%%%%%%%%%%%%%%%%%%%%%%%%%%%%%%%%%%
%
%\subsection{Marginal Advantage of Equality Tests}
%\label{sec: Marginal Advantage of Equal-to Tests}
%\input{05-5_marginal_advantage_of_equality_tests.tex}
%
%%%%%%%%%%%%%%%%%%%%%%%%%%%%%%%%%%%%%%%%%%%%%%%%%%%%%%%%%%%%%%%%%%%%%%%%%%%%%%%
%%%%%%%%%%%%%%%%%%%%%%%%%%%%%%%%%%%%%%%%%%%%%%%%%%%%%%%%%%%%%%%%%%%%%%%%%%%%%%%

Suppose that the heaviest key in $[i,j]$ is in a subinterval $[i',j'] \subsetneq [i,j]$.
And suppose that for $[i',j']$ the optimum choice is a less-than test. 
Does it necessarily imply that the optimum choice for $[i,j]$ is also a less-than test? 
Intuitively, that would show that the marginal advantage of equal-to tests
(by how much they are better than the best less-than test) decreases 
when the interval gets bigger. That's quite intuitive, since the relative weight
of this key to the total weight of the interval gets smaller, so it is less
likely to be used by the optimum tree in an equal-to test. (This is also consistent
with the threshold result in Theorem~\ref{theorem: lambda-}: If the weight of the
heaviest key becomes smaller than $\onefourth$ of the total interval weight, then
the tree uses a less-than-test.) We show that this intuition is not valid.

\paragraph{Counter-example.} Let $\alpha =11$, $\alpha' = 13$, and $\beta = 23$. Consider the following instance:

\smallskip
\begin{center}
	\begin{tabular}{|c|c|c|c|c|}\hline
		1 & 2 & 3 & 4 & 5   \\ \hline
		$\alpha$ & $\beta$ & $\beta$ & $\alpha'$ & $0$ \\ \hline
	\end{tabular}
\end{center}
\smallskip
For the smaller instance $[1,4]$ the optimum values for the equal-to test and the
less-than test are

\begin{equation*}
	C^=_{1,4} \;=\; \beta + 3\beta + 4 (\alpha + \alpha') \;=\; 141
	\quad\textrm{and}\quad
	C^<_{1,4} \;=\; C^<_{1,4} ( 2) \;=\; 2 ( 2\beta + \alpha + \alpha') \;=\; 140.
\end{equation*}
For the larger instance $[1,5]$ the optimum values for the equal-to test and the
less-than test are

\begin{equation*}
	C^=_{1,5} \;=\; \beta + 3\beta + 3\alpha' +  4\alpha  \;=\; 152
	\quad\textrm{and}\quad
	C^<_{1,5} \;=\; C^<_{1,5} ( 2) \;=\; 2 ( 2\beta + \alpha) + 3 \alpha' \;=\; 153.
\end{equation*}
So we have $C^=_{1,4} > C^<_{1,4}$, but $C^=_{1,5} < C^<_{1,5}$.

\begin{figure}[t]
	\centering
	\begin{subfigure}[b]{0.24\textwidth}
		\includegraphics[scale=0.8]{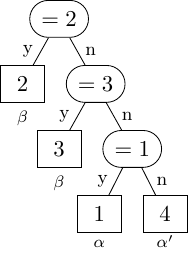}
		\caption{$C_{1, 4}^=$}
	\end{subfigure}
	\begin{subfigure}[b]{0.24\textwidth}
			\includegraphics[scale=0.8]{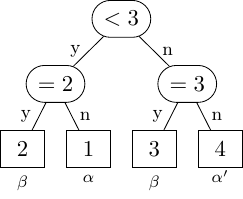}
			\caption{$C_{1, 4}^<$}
	\end{subfigure}
	\begin{subfigure}[b]{0.24\textwidth}
			\includegraphics[scale=0.8]{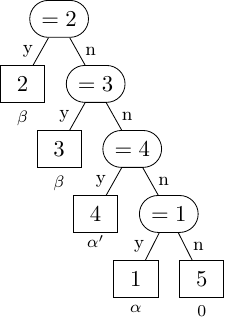}
			\caption{$C_{1, 5}^=$}
	\end{subfigure}
	\begin{subfigure}[b]{0.24\textwidth}
			\includegraphics[scale=0.8]{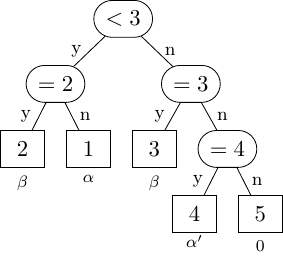}
			\caption{$C_{1, 5}^<$}
	\end{subfigure}
	\caption{Illustration of the counter-example in Section~\ref{sec: Marginal Advantage of Equality Tests}.}
	\label{fig: counterexample12}
\end{figure}

%%%%%%%%%%%%%%%%%%%%%%%%%%%%%%%%%%%%%%%%%%%%%%%%%%%%%%%%%%%%%%%%%%%%%%%%%%%%%%
%%%%%%%%%%%%%%%%%%%%%%%%%%%%%%%%%%%%%%%%%%%%%%%%%%%%%%%%%%%%%%%%%%%%%%%%%%%%%%

\subsection{Bounding Dynamic-Programming Minimizers by Weight}
\label{sec: bounding dynamic-programming minimizers by weight}
%%%%%%%%%%%%%%%%%%%%%%%%%%%%%%%%%%%%%%%%%%%%%%%%%%%%%%%%%%%%%%%%%%%%%%%%%%%%%%
%
%\subsection{Bounding Dynamic-Programming Minimizers by Weight}
%\label{sec: bounding dynamic-programming minimizers by weight}
%\input{05-6_bounding_dynamic-programming_minmimizers_by_weight.tex}
%
%%%%%%%%%%%%%%%%%%%%%%%%%%%%%%%%%%%%%%%%%%%%%%%%%%%%%%%%%%%%%%%%%%%%%%%%%%%%%%

Corollary~\ref{corollary: First Side Weight Threshold}  provides a potential algorithmic improvement as it extends \cite{Anderson2002}'s thresholds to less-than tests. More specifically it implies that the minimizer $l = L_{i, j}^h$ has to satisfy $\min\braced{w_{i, l}^h , w_{l+1, j}^h} \ge \onefourth w_{i,j}^h$. In other words, for any optimal tree rooted at a less-than test, the left and right branches of the root are relatively balanced, in the sense that the smaller branch is at least one fourth of the instance's total weight. A simple modification is to only search for minimizers where the side-weight is at least \(\frac{1}{4} w_{i, j}^h\).

For a given sub-problem \(I_{i, j}^h\), the ``refined interval'' \(\calI_{i, j}^h\) is defined as the set
	\[\calI_{i, j}^h = \braced{ i \le l < j \suchthat sw_{i, j}^h \parend{l} \ge \tfrac{1}{4} w_{i, j}^h}\]
where \(sw_{i, j}^h \parend{l} = \min\braced{w_{i, l}^h , w_{l+1 , j}^h}\).

From Corollary~\ref{corollary: First Side Weight Threshold}, $L_{i, j}^h$ must be in $\calI_{i, j}^h$. We can show that \(\calI_{i, j}^h\) is always an interval with no holes. By increasing the cut-point \(l\), the left-branch increases in weight and the right-branch decreases. So if \(l \in \calI_{i, j}^h\) but \(l+1 \notin \calI_{i, j}^h\), then the right branch is the side-branch of $l+1$, and moving the cut-point further right only decreases the side-weight as the right-branch gets smaller. Likewise happens for moving cut-points left. We can then represent each interval with end points \(\calI_{i, j}^h = [a_{i, j}^h , b_{i, j}^h]\). By pre-computing these intervals (or even computing them on an ad hoc basis in the recurrence) we can potentially narrow the number of sub-problems in the dynamic programming algorithm.

%%%%%%%%%%%%%%%%%%%%%%%

\begin{observation} \label{lemma: Computing Refined Intervals}
	For an instance of \(n\) keys, the set of refined intervals can be computed in \(O\parend{n^3 \log n}\) time.
\end{observation}

\begin{proof}
	The algorithm is straightforward. We first compute the weights of each sub-problem in \(O\parend{n^3}\) time in the standard fashion. Then for each sub-problem \(I_{i, j}^h\), we perform three rounds of binary search.
	
	First, we find one element \(c \in \calI_{i, j}^h\), if there is any. To do this, we select the midpoint \(m\) of \(\brackd{i, j-1}\) and check if \(sw_{i, j}^h \parend{m} \ge \frac{1}{4} w_{i, j}^h\). If yes, then we are done. If not, we check which branch of the cut-point \(m\) is lighter. If the left branch is lighter, we search in the interval \(\brackd{m+1, j-1}\), as we need to move the cut point right to make the branch heavier. Likewise, we search in the interval \(\brackd{i, m-1}\) if the right branch is lighter. We continue until we converge to a single cut-point, and if said point's side weight is still too light, then \(R_{i, j}^h\) is empty and we are done.
	
	If we found a cut-point \(c \in \calI_{i, j}^h\), we use binary search for each end point \(a_{i, j}^h\) and \(b_{i, j}^h\). To find the lower bound, we start searching in the interval \(\brackd{i, c}\). We select the midpoint \(m\) and check if \(sw_{i, j}^h \parend{m} \ge \frac{1}{4} w_{i, j}^h\). If no, then we search in \(\brackd{m+1, c}\). If yes, we then check if \(sw_{i, j}^h \parend{m-1} \ge \frac{1}{4} w_{i, j}^h\). If no (or if \(m-1 < i\)), then \(m\) is our lower bound and we are done. If yes, then search in \(\brackd{i, m-1}\). The search works similarly for the upper bound.
\end{proof}

Notably, refined intervals may still be large in proportion to the interval $\brackd{i, j-1}$ (if all keys have the same weight, then $\barred{\calI_{i, j}^h} = \half\parend{j-i+1}$), so they can't significantly cut down on the time it takes to find minimizers. However, they may significantly reduce the number of necessary sub-problems we need to compute. Consider a modified recurrence which uses refined intervals and also the two heaviest-key thresholds $\lowlambda$ and $\highlambda$:

\begin{align*}
	C_{i, j}^h &= 	w_{i, j}^h +
	\begin{cases}
		C_{i, j}^{h-1}, & \text{if } a_h \in I_{i, j}^h \text{ and } w_{a_h} \ge \threesevenths w_{i, j}^h ,\\
		S_{i, j}^h, & \text{if } a_h \in I_{i, j}^h \text{ and } w_{a_h} < \onefourth w_{i, j}^h ,\\
		\min\braced{C_{i, j}^{h-1}, S_{i, j}^h}, & \text{else.}
	\end{cases} \\
	S_{i, j}^h &= \min_{l \in \calI_{i, j}^h}\braced{C_{i, l}^h + C_{l+1, j}^h}.
\end{align*}

In this recurrence, we consider only one type of comparison test when $w_{a_h}$ hits a certain threshold and limit the cut-points we consider for less-than tests.

\paragraph{Counter-example.} Consider the $n$-key instance with weights

\begin{align*}
	(X_1 &, G, X_2, G, X_3, G, X_0), \quad \text{where} \\
	X_k &= \parend{\gamma^{4i+k}}_{i=0}^{n/7 -1} = \parend{\gamma^{k}, \gamma^{4 + k}, \ldots, \gamma^{n/7-1+k}},\\
	G &= \parend{0}_{i=1}^{n/7}.
\end{align*}

Visually, we have a geometric sequence distributed modulo $4$ into smaller sub-sequences, with some ``garbage'' $G$ in between. Here we'll fix $\gamma = \threefourths$. Then the total weight of this instance is $4 -\epsilon$, where $\epsilon$ can be made arbitrarily small by considering larger $n$. We'll show that the recurrence defined above creates $\Omega\parend{n^3}$ sub-problems that take $\Omega\parend{n}$ time to compute. We do this by considering ``rounds of decisions''. 

For this instance, the heaviest weight is $1$, slightly above $\onefourth$ the total. In which case the recurrence has to consider both equal-to tests and less-than tests. This case occurs even after taking several equal-to tests, so long as we don't take too many. And after taking four equal-to tests (one from each $X_i$), we have an identical instance, but with slightly fewer
keys, whose weights are scaled by $\gamma$. In our first round, remove some fraction of keys using equal-to tests to generate $\Omega\parend{n}$ sub-problems, each comparable to the original instance but scaled by some $\gamma^i$. The exact fraction of $n$ keys doesn't matter, so long as the heaviest weight is still approximately $\onefourth$ the total weight of each sub-problem.

In the second round, for each sub-problem from round 1, the keys corresponding to $X_0$ makes up about $1-\gamma^4 \approx 0.684$ of the total weight. This means all of the cut-points in the garbage $G$ between $X_0$ and $X_3$ are plausible cuts in the sub-problem's refined interval. Then for each round 1 sub-problems, the recurrence must consider at least $n/7$ sub-problems corresponding to less-than tests in this garbage region. This creates $\Omega\parend{n^2}$ round 2 sub-problems.

In the third round, for each sub-problem from round 2, weights are comprised of keys from $X_1$, $X_2$, and $X_3$. Factoring out the $\gamma^i$ term created from round 1, each sub-problem has heaviest weight $\gamma$ and approximate total weight $w \approx4 - 1/\parend{1-\gamma^4}$ . In which case the heaviest weight takes approximately $\gamma/w \approx 0.296$ of the total weight, still between $\onefourth$ and $\threesevenths$. So each sub-problem still needs to consider less-than tests. The keys in $X_1$ have weight about $\gamma/\parend{1-\gamma^4} \approx 0.432$ , so all of the keys in the garbage $G$ between $X_1$ and $X_2$ are in the refined interval. So each round 2 sub-problem generates at least $n/7$ sub-problems corresponding to these cut-points.

We now have $\Omega\parend{n^3}$ sub-problems, each containing keys from $X_2$ or $X_3$ surrounded by zero-weight keys from the leftmost and rightmost garbage regions. Ignoring the $\gamma^i$ term from round 1, the heaviest weight is $\gamma^2$ and total weight approximately $w = 4 - (1+\gamma)/(1-\gamma^4) $ . The heaviest weight is then about $\gamma / w \approx 0.391$ of the total weight, still between $\onefourth$ and $\threesevenths$. Then keys in $X_3$ have weight about $\gamma^2 / (1-\gamma^4)$ , making about $0.571$ of the total weight. So the garbage between $X_2$ and $X_3$ are cut-points in the refined interval, and thus each round 3 sub-problem needs to evaluate $n/7$ cut-points, taking $\Omega\parend{n}$ time to compute. Therefore this instance takes $\Omega\parend{n^4}$ time to compute.

%These refined intervals could in theory speed up computation significantly by decreasing the runtime of each sub-problem, while also reducing the number of sub-problems reached in the recursion. However, it seems like such sets are not enough to beat \(O\parend{n^4}\) on all instances. \todo{Add actual counter example here.}

%%%%%%%%%%%%%%%%%%%%%%%%%%%%%%%%%%%%%%%%%%%%%%%%%%%%%%%%%%%%%%%%%%%%%%%%%%%%%%
%%%%%%%%%%%%%%%%%%%%%%%%%%%%%%%%%%%%%%%%%%%%%%%%%%%%%%%%%%%%%%%%%%%%%%%%%%%%%%
%%%%%%%%%%%%%%%%%%%%%%%%%%%%%%%%%%%%%%%%%%%%%%%%%%%%%%%%%%%%%%%%%%%%%%%%%%%%%%

\section{Algorithms for Bounded Weights}
\label{sec: algorithms for bounded weights}
%%%%%%%%%%%%%%%%%%%%%%%%%%%%%%%%%%%%%%%%%%%%%%%%%%%%%%%%%%%%%%%%%%%%%%%%%%%%%%
%
%\section{Algorithms for Bounded Weights}
%\label{sec: algorithms for bounded weights}
%\input{06_algorithms_for_bounded_weights.tex}
%
%%%%%%%%%%%%%%%%%%%%%%%%%%%%%%%%%%%%%%%%%%%%%%%%%%%%%%%%%%%%%%%%%%%%%%%%%%%%%%

In this section we show that if the ratio between the largest and smallest weight is $R$
then the optimum can be computed in time $O(n^3\log (nR))$.
We can assume that the weights are from the interval $[1,R]$.

%%%%%%%%%%%%%%%

\myparagraph{Constant $R$.}
Let's start with the case when $R$ is constant.
If $j - i + 1 > 4R$ then the maximum weight in $[i,j]$ will be less than $\frac{1}{4}$th 
of the total weight of $I^n_{i,j}$ (the total weight of interval $[i,j]$). 
By Theorem~\ref{theorem: lambda+ lower bound} (see~\cite{Anderson2002}),
the optimum tree will use a less-than test.
In other words, for intervals $[i,j]$ with $j - i + 1 > 4R$
we do not need to consider instances with holes (the keys from $[i,j]$ that are not in the instance).
The optimum cost for this interval is $C_{i,j} = C^n_{i,j} = S^n_{i,j}$.
So the algorithm can do this: (1) compute all optimal solutions $C_{i,j}$
for the diagonals with $j - i + 1 \le 4R$ even by brute
force (as they have each constant size),
and then (2) for the remaining diagonals use the recurrence
$C_{i,j} = 
	\min_{l = i,...,j-1}\braced{ C_{i,l} + C_{l+1,j}}$.
This gives us an $O(n^3)$-time algorithm.

%%%%%%%%%%%%%%%

\myparagraph{Arbitrary $R$.}
We now sketch an algorithm for $R$ that may be a growing function of $n$. (We assume that $R \le c^n$ for some $c$. Otherwise 
this algorithm will not beat $O(n^4)$.) 

Consider an instance $I^h_{i,j}$. Recall that the keys in $[i,j]$ that are not in $I^h_{i,j}$ are called
\emph{holes}. All holes in $I^h_{i,j}$ are heavier than non-holes. 
The naive dynamic programming considers instances $I^h_{i,j}$ with any number of holes,
that is $0 \le |I^h_{i,j}| \le j-i+1$. The idea is that we only need to consider
holes that are sufficiently heavy.

Represent each sub-problem slightly differently: by its interval $[i,j]$ and the number of holes $s$.
Denote by $v^s$ the weight of the instance for interval $[i,j]$ from which $s$ holes have been
removed. We have that $v^0 \le nR$, and Theorem~\ref{theorem: lambda+ lower bound}
gives us that we only need to remove the $s$th hole if its weight is at least
$\onefourth$th of the total weight of the interval $[i,j]$ with $s-1$ holes already removed.
(Otherwise, the optimum tree for this sub-instance has the less-than test, so
the sub-instances with more holes in this interval do not appear in the overall optimal solution.)
This implies that $v^s \le \threefourths \cdot v^{s-1}$. 
If we thus remove $t$ holes, we will get $v^t \le (\threefourths)^t \cdot nR$.
But since all weights are at least $1$, this implies that $t = O(\log(nR))$.

Summarizing, we only need to consider sub-problems parametrized by $[i,j]$ and $s\le \log(nR)$.
This gives us $O(n^2\log(nR))$ instances and the running time
$O(n^3\log(nR))$.

%%%%%%%%%%%%%%%%%%%%%%%%%%%%%%%%%%%%%%%%%%%%%%%%%%%%%%%%%%%%%%%%%%%%%%%%%%%%%%
%%%%%%%%%%%%%%%%%%%%%%%%%%%%%%%%%%%%%%%%%%%%%%%%%%%%%%%%%%%%%%%%%%%%%%%%%%%%%%
%%%%%%%%%%%%%%%%%%%%%%%%%%%%%%%%%%%%%%%%%%%%%%%%%%%%%%%%%%%%%%%%%%%%%%%%%%%%%%

\bibliographystyle{plain}
\bibliography{search_trees}

%%%%%%%%%%%%%%%%%%%%%%%%%%%%%%%%%%%%%%%%%%%%%%%%%%%%%%%%%%%%%%%%%%%%%%%%%%%%%%
%%%%%%%%%%%%%%%%%%%%%%%%%%%%%%%%%%%%%%%%%%%%%%%%%%%%%%%%%%%%%%%%%%%%%%%%%%%%%%
%%%%%%%%%%%%%%%%%%%%%%%%%%%%%%%%%%%%%%%%%%%%%%%%%%%%%%%%%%%%%%%%%%%%%%%%%%%%%%

\end{document}